\documentclass[11pt]{article}

	\usepackage{a4,geometry}

\usepackage{graphicx}
\usepackage{amsmath,amssymb,amsthm,mathtools}
\usepackage{paralist}
\usepackage{bm}
\usepackage{xspace}
\usepackage{url}
\usepackage{fullpage, prettyref}
\usepackage{boxedminipage}
\usepackage{wrapfig}
\usepackage{ifthen}
\usepackage{color}
\usepackage{xcolor}
\usepackage{framed}
\usepackage{algorithmic,algorithm}
\usepackage[pagebackref,letterpaper=true,colorlinks=true,pdfpagemode=none,urlcolor=blue,linkcolor=blue,citecolor=violet,pdfstartview=FitH]{hyperref}
\usepackage{fullpage}

\usepackage{thmtools}
\usepackage{thm-restate}

\newtheorem{theorem}{Theorem}[section]

\newtheorem{lemma}[theorem]{Lemma}
\newtheorem{claim}[theorem]{Claim}

\newtheorem{definition}[theorem]{Definition}

\newtheorem{observation}[theorem]{Observation}

\newcommand{\ignore}[1]{}

\newcommand{\cD}{\mathcal{D}}

\newcommand{\cP}{\mathcal{P}}

\newcommand{\cU}{{\cal U}}

\newcommand{\R}{\mathbb R}
\newcommand{\N}{\mathbb N}

\newcommand{\eps}{\varepsilon}

\newcommand{\NN}{\mathbb{N}}

\newcommand{\Exp}{\EX}

\newcommand{\EX}{\hbox{\bf E}}

\newcommand{\Sec}[1]{\hyperref[sec:#1]{\S\ref*{sec:#1}}} 
\newcommand{\Eqn}[1]{\hyperref[eq:#1]{(\ref*{eq:#1})}} 
\newcommand{\Fig}[1]{\hyperref[fig:#1]{Fig.\,\ref*{fig:#1}}} 
\newcommand{\Tab}[1]{\hyperref[tab:#1]{Tab.\,\ref*{tab:#1}}} 
\newcommand{\Thm}[1]{\hyperref[thm:#1]{Theorem\,\ref*{thm:#1}}} 
\newcommand{\Fact}[1]{\hyperref[fact:#1]{Fact\,\ref*{fact:#1}}} 
\newcommand{\Lem}[1]{\hyperref[lem:#1]{Lemma\,\ref*{lem:#1}}} 
\newcommand{\Prop}[1]{\hyperref[prop:#1]{Prop.~\ref*{prop:#1}}} 
\newcommand{\Cor}[1]{\hyperref[cor:#1]{Corollary~\ref*{cor:#1}}} 
\newcommand{\Conj}[1]{\hyperref[conj:#1]{Conjecture~\ref*{conj:#1}}} 
\newcommand{\Def}[1]{\hyperref[def:#1]{Definition~\ref*{def:#1}}} 
\newcommand{\Alg}[1]{\hyperref[alg:#1]{Alg.~\ref*{alg:#1}}} 
\newcommand{\Ex}[1]{\hyperref[ex:#1]{Ex.~\ref*{ex:#1}}} 
\newcommand{\Clm}[1]{\hyperref[clm:#1]{Claim~\ref*{clm:#1}}} 

\usepackage{mathrsfs}
\usepackage{ifthen}
\usepackage{setspace}
\newcommand{\Obs}[1]{\hyperref[def:#1]{Observation~\ref*{obs:#1}}} 
\def\doctype{1}
\def\tsubmission{2}

\ifnum\doctype=\tsubmission
	\newcommand{\full}[1]{}
	\newcommand{\submit}[1]{#1}
\else
	\newcommand{\full}[1]{#1}
	\newcommand{\submit}[1]{}
\fi
\submit{\usepackage{times}}

\def\I{{\mathsf I}}

\def\MON{{\tt MON}}
\def\dist{{\sf dist}}

\def\depth{{\sf depth}}
\def\fii{{f^{(r)}}}

\def\fia{f^{(r)}_{|\a}}
\def\a{{\mathbf a}}

\def\v{{\mathbf v}}

\def\max{{\sf max}}
\def\TV{{\sf TV}}
\def\lca{{\sf lca}}
\def\val{{\sf val}}
\def\lleft{{\sf left}}
\def\rright{{\sf right}}
\def\const{120}

\newcommand{\B}{\mathbf{B}}
\newcommand{\be}{{\bf e}}
\newcommand{\bzero}{{\bf 0}}
\newcommand{\inter}[2]{\I^{#2}_{#1}}

\newcommand{\set}[1]{\{#1\}}
\newcommand{\distmon}[2]{\dist_{#2}(#1,\MON)}

\newcommand{\diff}{\Delta}
\newcommand{\frest}[2]{{#1}_{|#2}}
\newcommand{\slice}{S}

\newcommand\ist{{i^*}}

\newcommand{\cPd}[1]{\cP({#1})}
\def\fext{f_{\textrm{\tt ext}}}

\def\pdiext{d_{\textrm{\tt ext}}}
\def\pdiext{\pdi_{\textrm{\tt ext}}}
\def\cPdd{{\cP}}

\def\dext{\pdiext}
\def\cross{{\sf cr}}
\def\str{{\sf st}}
\def\pdi{{\mathfrak m}}
\def\cD{\mathscr{D}}
\def\cU{\mathscr{U}}
\def\cP{\mathscr{P}}
\newcommand{\hcd}{{\tt hcd}}
\renewcommand{\lca}{{\tt lca}}
\def\VG{\mathcal{G}_{\mathsf{viol}}}
\renewcommand{\epsilon}{\eps}
\def\gg{g_\psi}

\newcommand{\lev}[2]{L^{#1}_{\scriptscriptstyle #2}}
\newcommand{\dep}[2]{L^{#1}_{\scriptscriptstyle \geq #2}}
\newcommand{\abo}[2]{L^{#1}_{\scriptscriptstyle < #2}}
\newcommand{\below}[2]{\mu^{#1}_{\scriptscriptstyle \geq #2}}

\newcommand{\hard}[2]{g^{\scriptscriptstyle #1}_{\scriptscriptstyle #2}}

\newcommand{\hone}[2]{g^{\scriptscriptstyle #1}_{\scriptscriptstyle #2}}
\newcommand{\mappos}{{\Psi^{-1}}}
\newcommand{\gset}{\mappos}
\newcommand{\gcol}{\PPsi}
\colorlet{shadecolor}{blue!02}
\begin{document}
\title{Property Testing on Product Distributions: \\ Optimal Testers for Bounded Derivative Properties \submit{ \\ (Extended Abstract)\footnote{\em This is an extended abstract 
and with many missing details. A full version~\cite{paper}, with the same title, can be found in the ECCC and ArXiV. All proofs can be found in the full version.}}}
\author{
Deeparnab Chakrabarty\thanks{Microsoft Research, {\tt dechakr@microsoft.com}}
\and Kashyap Dixit\thanks{Pennsylvania State University, {\tt kashyap@cse.psu.edu}, supported in part by NSF Grants CCF-0964655 and CCF-1320814}
\and Madhav Jha\thanks{Sandia National Labs, Livermore, {\tt mjha@sandia.gov}}
\and C. Seshadhri\thanks{Sandia National Labs, Livermore, {\tt scomand@sandia.gov}\newline Sandia National Laboratories is a multi-program laboratory managed and operated by Sandia Corporation, a wholly owned subsidiary of Lockheed Martin Corporation, for the U.S. Department of Energy's National Nuclear Security Administration under contract DE-AC04-94AL85000.
}
}
\date{}
\maketitle
\thispagestyle{empty}
\begin{abstract}
The primary problem in property testing is to decide whether a given function satisfies a certain property, or is far
from any function satisfying it. This crucially requires a notion of distance between functions. The most prevalent notion 
is the Hamming distance over the {\em uniform} distribution on the domain. 
This restriction to uniformity is more a matter of convenience than of necessity, and it is important to investigate distances induced by more general distributions.
In this paper, we make significant strides in this direction. 
We give simple and optimal testers for {\em bounded derivative properties} over {\em arbitrary product distributions}.
Bounded derivative properties include fundamental properties such as monotonicity and Lipschitz continuity. 
Our results subsume almost all known results (upper and lower bounds) on monotonicity and Lipschitz testing.\smallskip
We prove an intimate connection between bounded derivative property testing and binary search trees (BSTs). We exhibit a tester whose query complexity is the sum of 
expected depths of optimal BSTs for each marginal. 
Furthermore, we show this sum-of-depths is also a lower bound. 
A fundamental technical contribution of this work is an {\em optimal dimension reduction theorem} for all bounded derivative properties, which relates the distance of a function from the property to the distance of restrictions of the function to random lines. Such a theorem has been elusive even for monotonicity for the past 15 years, and our theorem is an exponential improvement to the previous best known result.
\end{abstract}
\newpage
\setcounter{page}{1}
\section{Introduction}
The field of property testing 
formalizes the following problem: how many queries are needed to decide if a given function satisfies a certain property?
Formally, a \emph{property} $\cP$ is a subset of functions. A tester solves the relaxed membership problem of distinguishing functions
in $\cP$ from those `far' from $\cP$.  To formalize `far', one requires a notion of \emph{distance}, $\dist(f,g) \in [0,1]$, between functions. 
A function $f$ is \emph{$\eps$-far from $\cP$} if $\dist(f,g)\geq \eps$ for all functions $g \in \cP$. 
The notion of distance is central to property testing.  
The most prevalent notion of distance in the literature is the Hamming distance over the {\em uniform} distribution, that is,  $\dist(f,g) := \Pr_{x \sim \cU} [f(x) \neq g(x)]$, where $\cU$ is the uniform distribution over the input domain. 
But the restriction to uniformity is more a matter of convenience than of necessity, and it is important and challenging  to investigate distances induced by more general distributions.
This was already underscored in the seminal work of Goldreich et. al.~\cite{GGR98} who
``stressed that the generalization of property testing to arbitrary distributions"  is important for applications.
Nevertheless, a vast majority of results in property testing have focused  solely on the uniform distribution. 
In this paper we investigate property testing of functions defined over the hypergrid $[n]^d$ with respect to distances induced by {\em arbitrary product distributions}.
Product distributions over this domain form a natural subclass of general distributions where
each individual coordinate is an arbitrary distribution independent of the other coordinates. They arise in many applications; the following are a couple of concrete ones.
{\it Differential privacy:} Recent work on testing differential privacy~\cite{DiJh+13} involve product distributions over the domain $[n]^d$.
In this application, each domain point represents a database and each coordinate is a single individual's data. 
A distribution on databases is given by independent priors on each individual. The goal in~\cite{DiJh+13} is to distinguish private mechanisms from those that
aren't private on `typical' databases. 
{\it Random testing of hardware:} Given an actual silicon implementation of a  circuit,
it is standard practice for engineers to test it on a set of random instances. Coordinates represent entities like memory
addresses, data, control flow bits, etc. One chooses an independent but not identical distribution over each input 
to generate realistic set of test cases. There are specific commands in hardware languages like VHDL and Systemverilog~\cite{VHDL,verilog}
that specify such coordinate-wise distributions.
\medskip
From a theoretical perspective, the study of property testing over non-uniform distributions has mostly led to work on specific problems.
For uniform distributions, 
it is known that broad classes of algebraic and graphic properties are testable~\cite{AS05,AENS06,KaSu08,BSS}.
But little is known in this direction even for product distributions.
One reason for this may be aesthetics: a priori, one doesn't expect a succinct, beautiful answer for 
testing over an arbitrary product distribution.
In this paper we make significant strides in property testing under arbitrary product distributions. 
We give {\em simple}, {\em optimal} testers for the class of \emph{bounded derivative} properties. This class contains
(and is inspired by) the properties of monotonicity and Lipschitz continuity, which are of special interest in property testing.
In fact, the {\em same} tester works for all such properties. 
Furthermore, our `answer' is aesthetically pleasing:  the optimal query complexity with respect to a product distribution is the sum of optimal binary search tree depths over the marginals.
In particular, our results resolve a number of open problems in monotonicity testing, and subsume all previous 
upper and lower bounds over any product, including the uniform, distribution. 
\paragraph{Previous Work.}
We set some context for our work. The property of monotonicity is simple. There is a natural coordinate-wise partial order
over $[n]^d$. For a monotone function, $x \prec y$ implies $f(x) \leq f(y)$.
Monotonicity is one of the most well-studied properties in the area~\cite{EKK+00, GGLRS00,DGLRRS99,LR01,FLNRRS02,AC04,E04,HK04,PRR04,ACCL07,BRW05,BGJRW09,BCG+10,BBM11,ChSe13,ChSe13-2,BlJh+13}. 
A function is $c$-Lipschitz continuous if for all $x,y$, $|f(x) - f(y)| \leq c\|x-y\|_1$.
Lipschitz continuity is a fundamental mathematical property with applications to differential privacy and program robustness.
The study of Lipschitz continuity in property testing 
is more recent~\cite{JR11, AJMS12, ChSe13, DiJh+13, BlJh+13}.
With the exception of~\cite{HalevyK07, HK04,AC04, DiJh+13}, all the previous works are in the uniform distribution setting, for which the story is mostly clear:
there is an $O(\epsilon^{-1}d\log n)$-query tester for both properties~\cite{ChSe13}, and this is optimal for monotonicity~\cite{ChSe13-2}.
For Lipschitz continuity, an $\Omega(d\log n)$ non-adaptive lower bound has been proved~\cite{BlJh+13} recently.
For general product distributions, the story has been far less clear.
Ailon and Chazelle~\cite{AC04} design an $O(2^dH/\eps)$-query tester for monotonicity over product distributions, where $H$ is the Shannon entropy of the distribution.
This work connects property testing with information theory, and the authors explicitly ask whether the entropy is the ``correct answer''. 
There are no non-trivial lower bounds known for arbitrary product distributions.
For Lipschitz continuity, no upper or lower bounds are known for general hypergrids, although
an $O(d^2)$-query tester is known for the hypercube ($\{0,1\}^d$) domain~\cite{DiJh+13}. 
Halevy and Kushilevitz~\cite{HK04,HalevyK07} study monotonicity testing in the \emph{distribution-free setting}, where the tester does not know
the input distribution but has access to random samples. Pertinent to us, they show a lower bound of $\Omega(2^d)$ for monotonicity testing over \emph{arbitrary distributions}; for product distributions, Ailon and Chazelle~\cite{AC04} give an $O(\eps^{-1}d2^d\log n)$-query distribution-free tester.
\subsection{Bounded Derivative Properties}
To describe the class of bounded derivative properties, we first set some notation.
For an integer $k$, we use $[k]$ to denote the set $\{1,2,\ldots, k\}$.
Consider a function $f:[n]^d \mapsto \R$ and a dimension $r \in [d]$. Define $\partial_r f(x) := f(x+\be_r) - f(x)$, where $\be_r$ is the unit
vector in the $r$th dimension ($\partial_r f$ is defined only on $x$ with $x_r < n$.).
\begin{definition} \label{def:bound} An ordered set $\B$ of $2d$ functions 
$l_1, u_1, l_2, u_2, \ldots, l_d, u_d: [n-1] \mapsto \R$
is called a \emph{bounding family} if for all $r \in [d]$ and $y \in [n-1]$, $l_r(y) < u_r(y)$.
Let $\B$ be a bounding family of functions.
The property of being \emph{$\B$-derivative bounded}, denoted as $\cP(\B)$, is the set of functions
$f:[n]^d \mapsto \R$ such that: for all $r \in [d]$ and $x \in [n]^d$,
\begin{equation}
\label{eq:defnbnd}
l_r(x_r) \leq \partial_r f(x) \leq u_r(x_r).
\end{equation}
\end{definition}
This means the $r$th-partial derivative of $f$ is bounded by quantities that only
depend on the $r$th coordinate. Note that this dependence is completely arbitrary,
and different dimensions can have completely different bounds. 
This forms a rich class of properties which includes monotonicity and $c$-Lipschitz continuity.
To get monotonicity, simply set $l_r(y) = 0$ and $u_r(y) = \infty$ for all $r$. To get $c$-Lipschitz continuity,
set $l_r(y) = -c$ and $u_r(y) = +c$ for all $r$. 
The class also includes the property  
demanding monotonicity for some (fixed) coordinates and the $c$-Lipschitz continuity for others; and 
the non-uniform Lipschitz property that demands different Lipschitz constants for different coordinates. 
\begin{definition} \label{def:tester} 
Fix a bounding family $\B$ and product distribution $\cD = \prod_{r\leq d}\cD_r$.
Define $\dist_\cD(f,g) = \Pr_{x \sim \cD}[f(x) \neq g(x)]$. A property tester for $\cP(\B)$ with respect to $\cD$
takes as input proximity parameter $\eps > 0$ and has query access to function $f$. If $f \in \cP(\B)$, the tester accepts with probability $> 2/3$.
If $\dist_\cD(f,\cP(\B)) > \eps$, the tester rejects with probability $> 2/3$.
\end{definition}
\subsection{Main Results}
Our primary result is a property tester for all bounded-derivative properties over any product distribution.
The formal theorem requires some definitions of search trees.
Consider any binary search tree (BST) $T$ over the universe $[n]$, and let the depth of a node denote the number of {\em edges} from it to the root.
For a distribution $\cD_r$ over $[n]$, the {\em optimal BST for $\cD_r$} is the BST minimizing the expected depth of vertices drawn from $\cD_r$.
Let $\Delta^*(\cD_r)$ be this optimal depth:  a classic dynamic programming solution finds this optimal tree~\cite{Knuthvol3,YAO82} in polynomial time.
Given a product distribution $\cD = \prod_{r \leq d}\cD_r$, we abuse notation and let $\Delta^*(\cD)$ denote the sum $\sum_{r=1}^d \Delta^*(\cD_r)$.
\begin{restatable}{theorem}{distknown}
\label{thm:distknown-weak}\label{thm:distknown} {\bf [Main upper bound]}
Consider functions $f:[n]^d \mapsto \R$. Let $\B$ be a bounding family and $\cD$ be a product distribution.
There is a 
tester for $\cP(\B)$ w.r.t. $\cD$ making $100\eps^{-1}\Delta^*(\cD)$ queries.
\end{restatable}
The tester is non-adaptive with one-sided error, that is, the queries don't depend on the answers, and the tester always accepts functions satisfying the property. Furthermore, the {\em same} tester works for all bounding families, that is,  the set of queries made by the tester doesn't depend on $\B$.
Interestingly, the ``worst" distribution is the uniform distribution, where $\Delta^*(\cD)$ is maximized to $\Theta(d\log n)$.
We remark that the class of bounded derivative properties was not known to be testable even under uniform distributions. Results were known~\cite{ChSe13}  (only under
the uniform distribution) for the subclass where all $l_r$ (and $u_r$) are the same, constant function. 
To give perspective on the above result, it is instructive to focus on say just monotonicity (one can repeat this for Lipschitz).
Let $H(\cD)$ denote the Shannon entropy of distribution $\cD$ over the hypergrid.
It is well-known that $\Delta^*(\cD_r) \leq H(\cD_r)$ (see~\cite{Melhorn75} for a proof), so $\Delta^*(\cD)\leq H(\cD)$ for product distribution $\cD$. 
\begin{restatable}{corollary}{disthyp}\label{cor:ub-hg}\label{thm:ub-hg} Consider functions $f:[n]^d \mapsto \R$. Monotonicity testing over a product distribution $\cD$
can be done with $100 H(\cD)/\eps$ queries.
\end{restatable}
This is an \emph{exponential} improvement over the previous best result of Ailon and Chazelle~\cite{AC04},  who give a monotonicity tester with query complexity $O(2^d H(\cD)/\eps)$. 
Observe that for uniform distributions, $H(\cD) = \Theta(d\log n)$, and therefore the above result subsumes the optimal testers of~\cite{ChSe13}.
Now consider the monotonicity testing over the boolean hypercube. 
\begin{restatable}{corollary}{distknowncube}\label{cor:ub-hc}\label{cor:ub-cube} Consider functions $f:\{0,1\}^d \mapsto \R$. Monotonicity testing over any product distribution $\cD = \prod_{r=1}^d \cD_r$, where 
each $\cD_r = (\mu_r,1-\mu_r)$, can be done with $100\eps^{-1}\sum_{r=1}^d \min(\mu_r,1-\mu_r)$ queries.
\end{restatable}
Given that monotonicity testing over the hypercube has received much attention~\cite{GGLRS00,DGLRRS99,LR01,FLNRRS02,BBM11,ChSe13,ChSe13-2},
it is somewhat surprising that \emph{nothing non-trivial} was known even over the $p$-biased distribution for $p\neq 1/2$; our result implies an $O(\epsilon^{-1}pd)$-query tester.
The above corollary also asserts that entropy of a distribution {\em doesn't capture the complexity of monotonicity testing} since 
the entropy, $\sum_r \mu_r\log(1/\mu_r) + (1-\mu_r)\log(1/(1-\mu_r))$, can be larger than the 
query complexity described above by a logarithmic factor. For example, if each $\mu_r = 1/\sqrt{d}$, the tester of \Cor{ub-cube} requires $O(\sqrt{d}/\eps)$ queries,
while $H(\cD) = \Theta(\sqrt{d}\log d)$.
We complement \Thm{distknown} with a matching lower bound, cementing the connection between testing of bounded-derivative properties and optimal search tree depths.
This requires a technical definition of stable distributions, which is necessary for the lower bound. To see this
consider a distribution $\cD$ for which there exists a product distribution $\cD'$ such that $||\cD'-\cD||_{\TV} \leq \epsilon/2$ but $\Delta^*(\cD') \ll \Delta^*(\cD)$. 
One could simply apply \Thm{distknown} with $\cD'$ to obtain a tester with a much better query complexity than $\Delta^*(\cD)$. 
$\cD$ is called $(\epsilon',\rho)$-stable if $\|\cD-\cD\|\leq \epsilon'$ implies $\Delta^*(\cD')\geq\rho\Delta^*(\cD)$, for any product distribution $\cD'$.
\begin{restatable}{theorem}{mainlowerbound} {\bf [Main lower bound]}
\label{thm:the-lower-bound} 
For any parameter $\epsilon$, there exists $\epsilon' = \Theta(\epsilon)$ such that 
for any bounding family $\B$ and $(\eps',\rho)$-stable, product distribution $\cD$,
any (even adaptive, two-sided) tester for $\cP(\B)$ w.r.t. $\cD$ with proximity parameter $\eps$
requires $\Omega(\rho\Delta^*(\cD))$ queries.
\end{restatable}
This lower bound is new even for monotonicity testing over one dimension.  Ailon and Chazelle~\cite{AC04} explicitly ask for
lower bounds for monotonicity testing for domain $[n]$ over arbitrary distributions. Our upper and lower bounds
completely resolve this problem. For Lipschitz testing, the state of the art was a \emph{non-adaptive}
lower bound of $\Omega(d\log n)$ for the uniform distribution~\cite{BlJh+13}. Since the uniform distribution
is stable, the previous theorem implies an optimal $\Omega(d\log n)$ lower bound even for adaptive, two-sided testers over the uniform distribution. \smallskip
\noindent
The previous upper bounds are in the setting where the tester knows the distribution $\cD$. In the \emph{distribution-free} setting,
the tester only gets random samples from $\cD$ although it is free to query any point of the domain. 
As a byproduct of our approach, we also get results for this setting.
The previous best bound  was an $O(\eps^{-1}d2^{d}\log n)$ query tester~\cite{AC04}.
\begin{restatable}{theorem}{distfree}
\label{thm:main-free} Consider functions $f:[n]^d \mapsto \R$.
There is a distribution-free (non-adaptive, one-sided) tester for $\cP(\B)$ w.r.t. $\cD$
making $100\eps^{-1}d\log n$ queries.
\end{restatable}
\subsection{Technical highlights}\label{sec:ourtechs}
{\bf Optimal dimension reduction.} The main engine running the upper bounds is an 
optimal dimension reduction theorem. Focus on just the uniform distribution. 
Given $f:[n]^d \mapsto \R$ that is $\eps$-far from $\cP(\B)$, what is the expected distance of the function
restricted to a uniform random line in $[n]^d$?
This natural combinatorial question has been at the heart of various monotonicity testing results~\cite{GGLRS00,DGLRRS99,AC04,HK04}.
The best known bounds are that this expected distance is at least  
$\eps/(d2^d)$~\cite{AC04,HK04}. Weaker results are known for the Lipschitz property~\cite{JR11,AJMS12}.
We given an optimal resolution (up to constant factors) to this problem not only for the uniform distribution, but for any arbitrary product distribution,
and for any bounded derivative property.
In $[n]^d$, an $r$-line is a combinatorial line parallel to the $r$-axis.
Fix some bounding family $\B$ and product distribution $\cD = \prod_r \cD_r$. 
Note that $\cD_{-r} = \prod_{i \neq r} \cD_i$ is a distribution on $r$-lines.
If we restrict $f$ to an $r$-line $\ell$, we get a function $f|_\ell:[n] \mapsto \R$.
It is meaningful to look at the distance of $f|_\ell$ to $\cP(\B)$ (though this only
involves the bounds of $l_r, u_r \in \B$). 
Let $\dist^r_\cD(f,\cP(\B)) := \EX_{\ell \sim \cD_{-r}} [\dist_{\cD_r}(f|_\ell,\cP(\B))]$.
\begin{theorem}{\bf [Optimal Dimension Reduction]}\label{thm:dimred}
Fix bounding family $\B$ and product distribution $\cD$.
For any function $f$, $$\sum_{r=1}^d \dist^r_\cD(f,\cP(\B)) \geq \dist_\cD(f,\cP(\B))/4.$$
\end{theorem}
Let us give a short synopsis of previous methods used to tackle the case of monotonicity in the uniform distribution case.
For brevity's sake, let $\epsilon^r_f$ denote $\dist^r_\cU(f,\MON)$ and $\epsilon_f$ denote $\dist_\cU(f,\MON)$.
That is, $\epsilon^r_fn^d$ modifications makes the function monotone along the $r$-dimension, and the theorem above states
that $4\sum_r\epsilon^r_fn^d$ modifications suffice to make the whole function monotone. 
Either explicitly or implicitly, previous attempts have taken a constructive approach: they use the modifications along the $r$th dimensions to correct the whole function.
Although in principle a good idea, a bottleneck to the above approach is that correcting the function along one dimension may potentially 
introduce significantly larger errors along other dimensions. Thus, one can't just ``add up'' the corrections in a naive manner. The process is even more 
daunting when one tries this approach for the Lipschitz  property. 
Our approach is completely different, and is `non-constructive', and looks at all bounded-derivative properties in a uniform manner.
We begin by proving \Thm{dimred} for $\cP(\B)$ over the uniform distribution.
The starting point is to consider a weighted violation graph $G$, where any two domains point forming a violation to $\cP(\B)$
are connected (the weight is a ``magnitude" of violation). 
It is well-known that the size of a maximum matching $M$ in $G$ is at least $\eps_f n^d/2$.
The main insight is to use different matchings to get handles on the distance $\epsilon^r_f$ rather than using modifications that correct the function.
More precisely, we  construct a sequence of special matchings $M = M_0,M_1,\ldots, M_d = \emptyset$, 
such that 
the drop in size $|M_{r-1}| - |M_r|$ is at most $2\epsilon^r_fn^d$, which proves the above theorem. 
This requires 
structural properties on the $M_r$'s proven using the \emph{alternating path machinery} developed in~\cite{ChSe13}. 
What about a general product distribution $\cD$? Suppose we `stretch' every point in every direction proportional to its marginal.
This leads to a `bloated' hypergrid $[N]^d$ where each point in the original hypergrid corresponds to a high-dimensional cuboid.
By the obvious association of function values, one obtains a $\fext:[N]^d \mapsto \R$.
If $\cP(\B)$ is monotonicity, then it is not hard to show that $\dist_\cD(f) = \dist_\cU(\fext)$. 
So we can apply dimension reduction for $\fext$ over the uniform distribution and map it back to $f$ over $\cD$. 
However, such an argument breaks down for Lipschitz (let alone general $\B$), 
since $\dist_\cU(f')$ can be much smaller than $\dist_\cD(f)$. The optimal fix for $\fext$ could perform non-trivial changes within the cuboidal
regions, and this cannot be mapped back to a fix for the original $f$. This is where the generality of the bounded-derivative  properties
saves the day. For any $\B$ and $\cD$, we can define a new bounding family $\B_{\tt ext}$ over $[N]^d$, such that
$\dist_\cD(f,\cP(\B)) = \dist_\cU(\fext,\cP(\B_{\tt ext}))$. Now, dimension reduction is applied to $\fext$ for $\cP(\B_{\tt ext})$ over $\cU$
and translated back to the original setting. 
\smallskip
\noindent
{\bf Search trees and monotonicity.} An appealing aspect of our results is the tight connection between
optimal search trees over product distributions to bounded-derivative properties. 
The dimension reduction lemma allows us (for the upper bounds) to focus on just the line domain $[n]$. 
For monotonicity testing on $[n]$ over an arbitrary distribution $\cD$, Halevy and Kushilevitz gave an $O(\eps^{-1}\log n)$-query distribution free tester~\cite{HK04},
and Ailon and Chazelle gave an $O(\eps^{-1}H(\cD))$-query tester~\cite{AC04}.
Pretty much every single result for monotonicity testing on $[n]$ involves some analogue of binary search~\cite{EKK+00,BRW05,ACCL07,PRR04,HK04,AC04,BGJRW09}.
But we make this connection extremely precise. We show that {\em any} binary search tree can be used to get a tester 
with respect to an arbitrary  distribution, whose expected query complexity is the expected depth of the tree with respect to the distribution. 
This argument is extremely simple in hindsight, but it is a significant conceptual insight. Firstly, it greatly simplifies earlier results -- using the completely
balanced BST, we get an $O(\eps^{-1}\log n)$-distribution free tester; with the optimal BST, we get  $O(\eps^{-1}H(\cD))$-queries.
The BST tester along with the dimension reduction, provides a tester for $[n]^d$ whose running time can be better than $H(\cD)$ (especially for the hypercube).
Most importantly, optimal BSTs are a crucial ingredient for our lower bound construction.
\smallskip
\noindent
{\bf Lower Bounds for Product Distributions.} The first step to general lower bounds is a simple reduction
from monotonicity testing to any bounded-derivative property. Again, the reduction may seem trivial in hindsight,
but note that special sophisticated constructions were used for existing Lipschitz lower bounds~\cite{JR11,BlJh+13}.
For monotonicity, we use the framework developed in~\cite{E04,ChSe13-2} that allows us to focus on
comparison based testers.
The lower bound for $[n]$ uses a convenient near-optimal BST. 
For each level of this tree we construct a `hard'  non-monotone function, leading to (roughly) $\Delta^*(\cD)$
such functions in case of stable distributions. These functions have 
violations to monotonicity lying in `different regions' of the line,
and any bonafide tester must make a different query to catch each function.
In going to higher dimensions, we face a significant technical hurdle.
The line lower bound easily generalizes to the hypergrid {\em if each marginal distribution is individually stable.} However, this may not be the case -- there are stable product distributions
whose marginals are unstable. As a result, each dimension may give `hard' functions with very small distance.
Our main technical contribution is to show how to {\em aggregate} functions from various dimensions together to obtain hard functions for the hypergrid in such a way that
the distances add up. This is rather delicate, and is perhaps the most technical portion of this paper. 
In summary, we show that for stable distributions, the total search-tree depth is indeed the lower bound for testing monotonicity, and via the reduction mentioned above, for any bounded-derivative property.
\subsection{Other Related Work.}\label{sec:other_related_work}
Monotonicity testing has a long history, and we merely point the reader
to the discussions in~\cite{ChSe13, ChSe13-2}. The work on testing over non-uniform
distributions was performed in~\cite{HK04,HalevyK07,AC04}, the details of which
have been provided in the previous section. 
Goldreich et al.~\cite{GGR98} had already posed the question of testing properties of functions over non-uniform distributions,
and obtain some results for dense graph properties. A serious study of the role of distributions
was undertaken by Halevy and Kushilevitz~\cite{HalevyK07,HK04,HalevyK05,HalevyK08}, who
formalized the concept of distribution-free testing. (Refer to Halevy's thesis~\cite{Hal-thesis} for a comprehensive study.)
Kopparty and Saraf extend the classic linearity test to classes of distributions, including product distributions~\cite{KoSa08}.
Glasner and Servedio~\cite{GlasnerS09} and Dolev and Ron~\cite{DolevR11} give various upper and lower bounds for
distribution-free testers for various properties  over $\set{0,1}^n$.
Non-uniform distributions were also considered recently in the 
works of Balcan et al.~\cite{BalcanBBY12} and~\cite{GoldreichR13} which constrain the queries that can be made by the tester to samples drawn from the distribution. 
Recent work of Berman et al.~\cite{BeRaYa14} introduces property testing over $\ell_p$-distances. We believe work along these lines studying richer notions of distance is critical to the growth
of property testing.
\full{
\paragraph{Note to the reader.} The paper is rather long, although, we hope the extended introduction above will allow the reader to choose the order in which to peruse the paper.
We give a brief outline of remainder. In \Sec{abstract}, we define a particular {\em quasi-metric} corresponding to a bounding family $\B$ and give an equivalent 
definition of the bounded-derivative property with respect to it. This definition is convenient and will be the one used for the rest of the paper. This section must  be read next.
The dimension reduction theorem is presented in its full glory in \Sec{dimred}. In \Sec{line-ub}, we describe the tester when the domain is just the line, and the easy generalization to the hypergrid via dimension reduction is presented in \Sec{hg-ub}. For lower bounds, we prove the reduction to monotonicity in \Sec{bdp-to-mono}, and describe the approach to montonicity lower bounds in \Sec{monotone-lb-framework}. The hard families for the line is given in \Sec{lb-line}, for the hypercube in \Sec{lb-cube}, and the general hypergrid lowerbound is described in \Sec{lb-hypergrid}. 
}
\section{Quasimetric induced by a Bounding Family}\label{sec:abstract}
It is convenient to abstract out $\cP(\B)$ in terms of a \emph{metric-bounded property}. Such ideas was used in~\cite{ChSe13} to give a unified proof
for monotonicity and Lipschitz for the uniform distribution. The treatment here is much more general.
We define a quasimetric depending on $\B$ denoted by $\pdi(x,y)$.
\begin{definition} \label{def:dist} Given bounding family $\B$, construct the weighted directed hypergrid $[n]^d$, where 
all adjacent pairs are connected by two edges in opposite directions. The weight of $(x+\be_r,x)$ is $u_r(x_r)$ and the weight
of $(x,x+\be_r)$ is $-l_r(x_r)$. $\pdi(x,y)$ is the shortest path weight from $x$ to $y$.
\end{definition}
Note that $\pdi$ is asymmetric, can take negative values, and $\pdi(x,y) = 0$ does not necessarily imply $x=y$. 
For these reasons, it is really a possibly-negative-pseudo-quasi-metric, although we will refer to it simply as a metric in the remainder of the paper.
Since $\B$ is a bounding family, any cycle in the $[n]^d$ digraph has positive weight,
and $\pdi(x,y)$ is well-defined.
Therefore, a shortest path from $x$ to $y$ is given by the rectilinear path obtained by decreasing the coordinates $r$ with $x_r > y_r$ and increasing the coordinates $r$ with $x_r< y_r$. 
A simple calculation yields
\begin{equation}
\label{eq:supergeneralLip}
\pdi(x,y) :=  \sum_{r:x_r > y_r} \sum_{t = y_r}^{x_r-1}\! u_r(t)  -  \sum_{r:x_r < y_r}\sum_{t = x_r}^{y_r-1}\!l_r(t)
\end{equation}
If a function $f\in \cP(\B)$, then applying \Eqn{defnbnd} on every edge of the path described above (the upper bound when we decrement a coordinate and the lower bound when we increment a coordinate), we get $f(x) - f(y) \leq \pdi(x,y)$ for any pair $(x,y)$. Conversely, if $\forall x,y, f(x) - f(y) \leq \pdi(x,y)$, then considering neighboring pairs gives $f\in \cP(\B)$.
This argument is encapsulated in the following lemma.
\begin{lemma} \label{lem:dist} $f \in \cP(\B)$ iff $~\forall x, y \in [n]^d$, $f(x) - f(y) \leq \pdi(x,y)$.
\end{lemma}
\noindent
When $\cP(\B)$ is monotonicity, $\pdi(x,y) = 0$ if $x \prec y$ and $\infty$ otherwise. For the $c$-Lipschitz property,
$\pdi(x,y) = c\|x-y\|_1$.
The salient properties of $\pdi(x,y)$ are documented below and can be easily checked\submit{ (proof in full version)}. 
\begin{lemma}\label{lem:dist-prop} $\pdi(x,y)$ satisfies the following properties.
\begin{asparaenum}
\item {\em (Triangle Inequality.)} For any $x,y,z$, $\pdi(x,z) \leq \pdi(x,y) + \pdi(y,z)$.
\item {\em (Linearity.)} If $x,y,z$ are such that for every $1\leq r\leq d$, either $x_r \leq y_r \leq z_r$ or $x_r \geq y_r \geq z_r$, then 
$\pdi(x,z) = \pdi(x,y) + \pdi(y,z)$.
\item {\em (Projection.)} Fix any dimension $r$. Let $x,y$ be two points with $x_r = y_r$. Let $x'$ and $y'$
be the projection of $x, y$ onto some other $r$-hyperplane. That is, 
$x'_r = y'_r$, and $x'_j = x_j$, $y'_j = y_j$  for $j\neq r$. Then, $\pdi(x,y) = \pdi(x',y')$ and $\pdi(x,x') = \pdi(y,y')$.
\end{asparaenum}
\end{lemma}
\full{
\begin{proof} $\pdi(x,x) = 0$ follows since the RHS of \Eqn{supergeneralLip} is empty. 
Triangle inequality holds because $\pdi(x,y)$ is a shortest path weight.
Linearity follows by noting $\sum_{t = y_r}^{x_r - 1} u_r(t) = \sum_{t = y_r}^{z_r - 1} u_r(t) + \sum_{t = z_r}^{x_r - 1} u_r(t)$. For projection, note that if $x_r = y_r$, the RHS of \Eqn{supergeneralLip}
has no term corresponding to $r$. Thus, $\pdi(x,y) = \pdi(x',y')$. Suppose $x'_r > x_r$. Then, $\pdi(x,x') = \sum_{t = x_r}^{x'_r} u_i(t)$ $= \pdi(y,y')$. A similar proof holds when $x'_r < x_r$.
\end{proof}
}
Henceforth, all we need is \Lem{dist} and \Lem{dist-prop}. We will interchangably use the terms $\cP(\B)$ and $\cP(\pdi)$ where $\pdi$ is as defined in \Eqn{supergeneralLip}.
In fact, since $\B$ and therefore $\pdi$ will be fixed in most of our discussion, we will simply use $\cP$ including the parametrization wherever necessary.
\begin{definition}[{\bf Violation Graph}]
The violation graph of a function $f$ with respect to property $\cPdd$, denoted as $\VG(f,\cPdd)$, has $[n]^d$ as vertices, and edge $(x,y)$ if it forms a violation to $\cPdd$, that is either $f(x) - f(y) > \pdi(x,y)$ or $f(y) - f(x) > \pdi(y,x)$.
\end{definition}
The triangle inequality of $\pdi$ suffices to prove the following version
of a classic lemma~\cite{FLNRRS02} relating the distance of a function to $\cP$ to the vertex cover of the violation graph.
\begin{lemma}\label{lem:characterization}\label{lem:matchmatch}
For any distribution $\cD$ on $[n]^d$, any bounded-derivative property $\cP$, and any function $f$,  $\dist_\cD(f,\cP) = \min_X \mu_\cD(X)$ where the minimum is over all vertex covers of $\VG(f,\cP)$. Thus, if $M$ is {\em any} maximal matching in $\VG(f,\cP)$, then for the {\em uniform distribution}, $|M|\geq \dist_\cU(f,\cP)n^d/2$.
\end{lemma}
\section{The Dimension Reduction Theorem}\label{sec:dimred}
For any combinatorial line $\ell$ in $[n]^d$, $f|_\ell:[n] \mapsto \R$ is $f$ restricted to $\ell$.
It is natural to talk of $\cP$ 
for any restriction of $f$, so
$\dist_{\cD_r}(\frest{f}{\ell}, \cP)$ is well-defined for any $r$-line $\ell$.
For any $1\leq r\leq d$, define the $r$-distance of the function:
\begin{equation}
\dist^r_\cD(f,\cP):=\Exp_{\ell \sim \cD_{-r}} [\dist_{\cD_r}(\frest{f}{\ell}, \cP)]
\end{equation} 
Call a function $f$ $r$-good if there are no violations along $r$-lines, that is, for any $x$ and $y$ on the same $r$-line, we have $f(x) - f(y) \leq \pdi(x,y)$. 
Observe that $\dist^r_\cD(f,\cP)$
is the minimum $\mu_\cD$-mass of points on which $f$ needs to be modified to make it $r$-good.
The following is the optimal dimension reduction theorem which connects the $r$-distances to the real distance.
\begin{theorem}[Dimension Reduction]\label{thm:dimred}
For any function $f$, any bounded-derivative property $\cP$, 
and any product distribution $\cD = \prod_{1\leq r\leq d}\cD_i$, \submit{$\sum_{r=1}^d \dist^r_\cD(f,\cP) \geq \dist_\cD(f,\cP)/4.$}
\full{$$\sum_{r=1}^d \dist^r_\cD(f,\cP) \geq \dist_\cD(f,\cP)/4.$$}
\end{theorem}
(It can be easily shown that $\sum_{r=1}^d \dist^r_\cD(f,\cP) \leq \dist_\cD(f,\cP)$, by simply
putting the same 1D function of all, say, $1$-lines.)
We first prove the above theorem for the uniform distribution. 
Recall the violation graph $\VG(f,\cP)$ whose edges are violation to $\cP$. We define weights on the edges $(x,y)$.
\begin{equation}
\label{eq:wtdefn}
w(x,y) := \max(f(x) - f(y) - \pdi(x,y), f(y) -f(x) - \pdi(y,x))
\end{equation}
Note that $w(x,y) > 0$ for all edges in the violation graph. Let $M$ be a maximum weight matching of minimum cardinality (MWmC). (Introduce an arbitrary tie-breaking rule to ensure this is unique.)
A pair $(x,y)\in M$ is an {\em $r$-cross pair} if $x_r \neq y_r$. 
The following theorem \full{(proof defered to \Sec{noviol})} establishes the crucial structural result about these MWmC matchings in violated graphs of $r$-good functions.
\begin{restatable}[No $r$-violations $\Rightarrow$ no $r$-cross pairs]{theorem}{noviol}
\label{thm:noviol}
Let $f$ be an $r$-good function. Then there exists an MWmC matching $M$ in $\VG(f,\cP)$ with no $r$-cross pairs.
\end{restatable}
\submit{The proof of \Thm{noviol} follows the {\em alternating path} methodology invented in~\cite{ChSe13}.
The challenge is that \emph{maximal} matchings can contain $r$-cross pairs, so we have to understand
how the maximum weight constrains the matching.
Consider an MWmC matching $M$ minimizing the number of $r$-cross pairs. Suppose, for contradiction's sake, there exists an $r$-cross pair $(x,y)$ with $x_r = a < b = y_r$. We consider the alternating
paths formed by two matchings: one is $H$ formed by pairs $(u,v)$ with $u_r = a, v_r = b$ and $u_j = v_j$ for $j\neq r$, and, the other is $\str(M)\subseteq M$, the pairs in $M$ that are not $r$-cross pairs. 
If the alternating path starting at $y$ does not contain a violating $H$-pair, we show that the matching $M$ can be improved either in weight, or in cardinality, or in number of $r$-cross pairs. 
So the alternating path contains a violating $H$-pair, contradicting the $r$-goodness of $f$.
This may sound simple, but notice that the weight of a pair $(u,v)$ is defined in \Eqn{wtdefn} as the maximum of two quantities. To argue about the weight of matchings, we need to know  which of the two is indeed the correct weight, otherwise the number of possibilities blows up exponentially.
In fact, we can {\em exactly} pin this down for the edges on the alternating path via an inductive argument.
This crucially uses all the three properties of $\pdi$ defined in \Lem{dist-prop}. Describing all this rigorously takes some space, and therefore has been omitted from this extended abstract. 
}
We proceed with the proof of \Thm{dimred} over the uniform distribution starting with some definitions.
\begin{definition}[{\em Hypergrid slices}]
Given an $r$-dimensional vector $\a\in [n]^r$, the {\em $a$-slice} is $\slice_\a := \{x\in [n]^d: x_j = \a_j, ~1\leq j\leq r\}$.
\end{definition}
\noindent
Each $\a$-slice is a $(d-r)$-dimensional hypergrid, and the various $\a$-slices for $\a\in[n]^r$ partition $[n]^d$.
Let $\frest{f}{\a}$ denote the restriction of $f$ to the slice $\slice_\a$. 
For two functions $f,g$ we use $\diff(f,g) := |\{x: f(x) \neq g(x)\}| = \dist_\cU(f,g)\cdot n^d$. 
The following claim relates the sizes of MWmC matchings to $\diff(f,g)$.
\begin{claim}\label{clm:matchdiff}
Let $f,g: [n]^d \mapsto \R$. Let $M$ and $N$ be the MWmC matchings in the violation graphs for $f$ and $g$, respectively. Then, $||M|-|N|| \leq \diff(f,g)$.
\end{claim}
\begin{proof}
The symmetric difference of $M$ and $N$ is a collection of alternating paths and cycles. $||M|-|N||$ is at most the number of alternating paths. 
Each alternating path must contain a point at which $f$ and $g$ differ, for otherwise we can improve either $M$ or $N$, either in weight or cardinality.
\end{proof}
Define a sequence of $d+1$ matchings $(M_0,M_1,\ldots,M_d)$ in $\VG(f,\cP)$ in non-increasing order of cardinality as follows.
 For $0 \leq r\leq d$, $M_r$ is the MWmC matching in $\VG(f,\cP)$
among matchings that {\em do not contain any $i$-cross pairs for $1\leq i\leq r$}. 
By \Lem{matchmatch}, we have $|M_0| \geq \dist_\cU(f,\cP)n^d/2$.
The last matching $M_d$ is empty and thus has cardinality $0$. 
\begin{lemma}\label{lem:mi}
For all $1\leq r\leq d$, we have $|M_{r-1}| -|M_{r}| \leq 2\cdot \dist^r_\cU(f,\cP)\cdot n^d$.%
\end{lemma}
\noindent
Adding the inequalities in the statement of \Lem{mi} for all $r$, we get $\dist_\cU(f,\cP)n^d/2  \leq|M_0| - |M_d| \leq 2\sum_{r=1}^d \dist^r_\cU(f,\cP)\cdot n^d.$
\noindent
This completes the proof of \Thm{dimred} for the uniform distribution. Now we prove \Lem{mi}.
\begin{proof}
Since $M_{r-1}$ has no $j$-cross pairs for $1\leq j\leq r-1$, all pairs of $M_{r-1}$ have both endpoints in the same slice $\slice_\a$ for some $\a\in [n]^{r-1}$. Thus, $M_{r-1}$ partitions into sub-matchings in each $\slice_\a$. Let $M^\a_{r-1}$ be the pairs of $M_{r-1}$ with both endpoints in slice $\slice_\a$, so
$|M_{r-1}| = \sum_{\a\in [n]^{r-1}} |M^\a_{r-1}|$.
Similarly, $M^\a_r$ is defined. Since $M_r$ has no $r$-cross pairs either, $\forall \a\in [n]^{r-1}$,
$|M^\a_{r}| = \sum_{i=1}^n |M^{(\a\circ i)}_{r}|$,
where $(\a\circ i)$ is the $r$-dimensional vector obtained by concatenating $i$ to the end of $\a$.
Observe that
for any $\a\in [n]^{r-1}$, $M^\a_{r-1}$ is an MWmC matching in $\slice_\a$ w.r.t. $\frest{f}{\a}$. Furthermore, for any $i\in [n]$, $M^{(\a\circ i)}_r$ is an MWmC matching in $\slice_{(\a\circ i)}$ w.r.t. $\frest{f}{(\a\circ i)}$.
\noindent
Let $\fii$ be the closest function to $f$ with no violations along dimension $r$. 
By definition, $\diff(f,\fii) = \dist^r(f,\cP)\cdot n^d$.
Now comes the crucial part of the proof. 
Fix $\a\in [n]^{r-1}$ and focus on the $\a$-slice $\slice_\a$.  Since $\fii$ has no violations along the $r$-lines, 
neither does $\frest{f^{(r)}}{\a}$. By \Thm{noviol}, there exists an MWmC matching $N^\a$ in $\slice_\a$ w.r.t. $\frest{f^{(r)}}{\a}$ which
has no $r$-cross pairs.
Therefore, $N^\a$ partitions as $N^\a = \bigcup_{i=1}^n N^{(\a\circ i)}$. Furthermore, each matching $N^{(\a\circ i)}$ is an MWmC matching in $\slice_{(\a\circ i)}$ with respect to the weights corresponding to the function  $f^{(r)}_{|(\a\circ i)}$.
Since $M^\a_{r-1}$ is an MWmC matching w.r.t. $\frest{f}{\a}$ and $N^\a$ is an MWmC matching w.r.t. $\fia$ in $\slice_\a$, \Clm{matchdiff} gives 
\begin{equation}\label{eq:3}
|N^\a| \geq |M^\a_{r-1}| - \diff(\frest{f}{\a},\frest{f^{(r)}}{\a})
\end{equation}
Since $M^{(\a\circ i)}_r$ is an MWmC matching w.r.t. $f_{|(\a\circ i)}$ and  $N^{(\a\circ j)}$ is an MWmC matching  w.r.t. $f^{(r)}_{|(\a\circ i)}$ in $\slice_{(\a\circ i)}$, \Clm{matchdiff} gives us 
$|M^{(\a\circ i)}_r| \geq |N^{(\a\circ i)}| - \diff(\frest{f}{(\a\circ i)}, \frest{f^{(r)}}{(\a\circ i)})$. Summing over all $1\leq i\leq n$,
\begin{equation}\label{eq:4}
|M^\a_r| \geq |N^\a| - \diff(\frest{f}{\a},\frest{f^{(r)}}{\a})
\end{equation}
Adding \Eqn{3}, \Eqn{4} over all $\a\in [n]^{r-1}$,
$|M_r| \geq |M_{r-1}| - 2\sum_{\a\in [n]^{r-1}}\diff(\frest{f}{\a},\frest{f^{(r)}}{\a}) $ $= |M_{r-1}| - 2\cdot\dist^r(f,\cP)\cdot n^d$.
\end{proof}
\subsection{Reducing from arbitrary product distributions}\label{sec:uniftogen}
\full{We reduce arbitrary product distributions to uniform distributions on what we call the bloated hypergrid.}
\submit{We just provide a high level sketch of the argument with details in the full version.}
Assume without loss of generality that all $\mu_{\cD_r}(j) = q_r(j)/N$, for some integers $q_r(j)$ and $N$.
Consider the $d$-dimensional $N$-hypergrid $[N]^d$. There is a natural many-to-one mapping from  $\Phi: [N]^d\mapsto [n]^d$ defined as follows. First fix a dimension $r$. 
Given an integer $1\leq t \leq N$, let $\phi_r(t)$ denote the index $\ell \in [1,n]$ such that 
$\sum_{j <\ell} q_r(j) < t \leq \sum_{j\leq \ell} q_r(j)$. That is, partition $[N]$ into $n$ contiguous segments of lengths $q_r(1),\ldots,q_r(n)$. Then $\phi_r(t)$ is the index of the segment where $t$ lies. The mapping $\Phi: [N]^d\mapsto [n]^d$ is defined as \submit{$\Phi(x_1,x_2\ldots,x_d) = \left(\phi_1(x_1), \phi_2(x_2), \ldots, \phi_\pdi(x_d) \right)$. The preimage of a point in $[n]^d$ is a 
`cuboid' in $[N]^d$.}
\full{\[\Phi(x_1,x_2\ldots,x_d) = \left(\phi_1(x_1), \phi_2(x_2), \ldots, \phi_\pdi(x_d) \right).\]
We use $\Phi^{-1}$ to define the set of preimages, so $\Phi^{-1}$ maps a point in $[n]^d$ to a `cuboid' in $[N]^d$. 
Observe that for any $x\in [n]^d$,
\begin{equation}\label{eq:obs}
|\Phi^{-1}(x)| = N^d\prod_{r=1}^d \mu_{\cD_r}(x) = N^d\mu_\cD(x).
\end{equation}
\begin{claim}\label{clm:XtoZ}
For any set $X\subseteq [n]^d$, define $Z \subseteq [N]^d$ as $Z := \bigcup_{x\in X}\Phi^{-1}(x)$. Then $\mu_\cD(X) = \mu_\cU(Z)$.
\end{claim}
\begin{proof}
The set  $Z = \bigcup_{x\in X}\Phi^{-1}(x)$ is the union of all the preimages of $\Phi$ over the elements of $X$. Since preimages are disjoint, we get
$|Z| = \sum_{x\in X}|\Phi^{-1}(x)| = N^d\mu_\cD(X)$. Therefore, $\mu_\cU(Z) = \mu_\cD(X)$.
\end{proof}
}
Given $f:[n]^d \mapsto \R$, we define its extension $\fext:[N]^d \mapsto \R$: 
\begin{equation}
\label{eq:fext}\fext(x_1,\ldots,x_d) = f(\Phi(x_1,\ldots,x_d)). 
\end{equation}
Thus, $\fext$ is constant on the cuboids in the bloated hypergrid corresponding to a point in the original hypergrid. 
\submit{Define the following metric on $[N]^d$: $\pdiext(x,y) =\pdi(\Phi(x),\Phi(y))$.
It is relatively easy to argue that $\pdiext$ also satisfies the conditions of \Lem{dist-prop}, and that
$\dist_\cD(f,\cPd{\pdi}) = \dist_\cU(\fext,\cP(\pdiext))$. So we can apply
the dimension reduction for $\fext$ for over the uniform distribution, and reverse the mapping
to get the dimension reduction for $f$ over $\cD$.
}
\full{
Define the following metric on $[N]^d$. 
\begin{equation}\label{eq:dext}
\textrm{For $x,y \in [N]^d$,} \quad \pdiext(x,y) =\pdi(\Phi(x),\Phi(y)) 
\end{equation}
The following statements establish the utility of the bloated hypergrid, and the proof of the dimension reduction of $f$ over $[n]^d$ w.r.t. $\cD$ follows easily from these and the proof for the uniform distribution.
\begin{lemma} \label{lem:distconsistency}
If $\pdi$ satisfies the conditions of \Lem{dist-prop} over $[n]^d$, then so does
$\pdiext$ over $[N]^d$.
\end{lemma}
\begin{proof} Consider $x,y,z \in [N]^d$.
Triangle inequality and well-definedness immediately follow from the validity of $\pdi$. 
Now for linearity. If $x_r\leq y_r\leq z_r$, then so is $\phi_r(x_r) \leq \phi_r(y_r)\leq \phi_r(z_r)$. Thus, $\Phi(x), \Phi(y),\Phi(z)$ satisfy linearity w.r.t. $\pdi$.
So, 
$\pdiext(x,z) = \pdi(\Phi(x),\Phi(z)) = \pdi(\Phi(x),\Phi(y)) + \pdi(\Phi(y),\Phi(z)) = \pdiext(x,y) + \pdiext(y,z)$.
Now for projection. Suppose  $x_r = y_r$ and $x'_r = y'_r$. Note that $\Phi(x)$ and $\Phi(y)$ have same $r$th coordinate,
and so do $\Phi(x')$ and $\Phi(y')$. 
Furthermore, $\Phi(x')$ (resp. $\Phi(y')$) is the projection of $\Phi(x)$ (resp. $\Phi(x)$). Thus we get $\pdiext(x,y) = \pdi(\Phi(x),\Phi(y)) = \pdi(\Phi(x'),\Phi(y')) = \pdiext(x',y')$, and similarly $\pdiext(x,x') = \pdiext(y,y')$.
\end{proof}
\begin{theorem}\label{thm:bhg}
$\dist_\cD(f,\cPd{\pdi}) = \dist_\cU(\fext,\cP(\pdiext))$.
\end{theorem}
\begin{proof}
($\geq$).
Let  $X\subseteq [n]^d$ be a vertex cover in $\VG(f,\cPd{\pdi})$ minimizing $\mu_\cD(X)$. From \Lem{characterization}, $\dist_\cD(f,\cPd{\pdi}) = \mu_\cD(X)$.
We claim $Z = \bigcup_{x\in X}\Phi^{-1}(x)$ is a vertex cover of $\VG(\fext,\cP(\pdiext))$. This implies $\dist_\cU(\fext,\cP(\pdiext)) \leq \mu_\cU(Z) = \mu_\cD(X) = \dist_\cD(f,\cPd{\pdi})$, where the first equality follows from  \Clm{XtoZ}.
Consider a violated pair $(u,v)$ in this graph and so wlog $\fext(u) - \fext(v) > \pdiext(u,v)$. Hence,
$f(\Phi(u)) - f(\Phi(v)) > \pdi(\Phi(u),\Phi(v))$ implying  $(\Phi(u),\Phi(v))$ is an edge in $\VG(f,\cPd{\pdi})$. Thus, either $\Phi(u)$ or $\Phi(v)$ lies in $X$ implying either $u$ or $v$ lies in $Z$.\smallskip
\noindent
($\leq$).
Let $Z \subseteq [N]^d$ be a vertex cover in $\VG(\fext,\cP(\pdiext))$ minimizing $\mu_\cU(Z)$. Therefore, $\dist_\cU(\fext,\cP(\pdiext)) = \mu_\cU(Z)$.
Define $X\subseteq [n]^d$ as $X = \{x\in [n]^d: \Phi^{-1}(x)\subseteq Z\}$. Therefore, $Z \supseteq \bigcup_{x\in X}\Phi^{-1}(x)$ and from \Clm{XtoZ} we get $\mu_\cU(Z)\geq \mu_\cD(X)$.
It suffices to show that $X$ is a vertex cover of $\VG(f,\cPd{\pdi})$. 
Consider a violated edge $(x,y)$ in this graph such that $f(x) - f(y) > \pdi(x,y)$. 
Suppose neither $x$ nor $y$ are in $X$.
Hence, there exists $u\in \Phi^{-1}(x) \setminus Z$ and $v\in \Phi^{-1}(y) \setminus Z$. 
So $\fext(u) - \fext(v) = f(\Phi(u)) - f(\Phi(v)) = f(x) - f(y) > \pdi(x,y) = \pdiext(\Phi(u),\Phi(v))$,
implying $(u,v)$ is a violation in $\VG(\fext,\cP(\pdiext))$. This contradicts the fact that $Z$ is a vertex cover.
\end{proof}
}
\full{
Fix a dimension $r$ and $r$-line $\ell$. Abusing notation, let $\Phi^{-1}(\ell)$ denote the collection of $r$-lines in $[N]^d$ 
that are mapped to $\ell$ by $\Phi$.
Note that $|\Phi^{-1}(\ell)| = N^{d-1} \mu_{\cD_{-r}}(\ell)$.
A proof identical to one above yields the following theorem.
\begin{theorem}
For any $r$-line, $\dist_{\cD_r}(\frest{f}{\ell},\cPd{\pdi}) = \dist_{\cU_r}(\frest{\fext}{\ell'},\cP(\pdiext))$ for all $\ell' \in \Phi^{-1}(\ell)$.
\end{theorem}
\noindent
Now we can complete the proof of \Thm{dimred}.
\begin{eqnarray*}
\dist^r_\cD(f,\cPd{\pdi}) & = & \sum_{\textrm{$r$-line $\ell$}} \mu_{\cD_{-r}}(\ell) \cdot \dist_{\cD_r}(\frest{f}{\ell},\cPd{\pdi}) \\
& = & \frac{1}{N^{d-1}} \sum_{\textrm{$r$-line $\ell$}} |\Phi^{-1}(\ell)| \cdot \dist_{\cD_r}(\frest{f}{\ell},\cPd{\pdi}) \\
& = & \frac{1}{N^{d-1}}\sum_{\textrm{$r$-line $\ell$}} \sum_{\ell' \in \Phi^{-1}(\ell)} \dist_{\cU_r}(\frest{\fext}{\ell'},\cP(\pdiext)) \\
&  = & \Exp_{\ell' \sim \cU_{-r}} [\dist_{\cU_r}(\frest{\fext}{\ell'},\cP(\pdiext))]
= \dist^r_{\cU}(\fext,\cPd{\dext}).
\end{eqnarray*}
We can apply the dimension reduction to $\fext$ for property $\cP(\pdiext)$ over the uniform distribution.
The proof of \Thm{dimred} for $f$ follows directly.}
\full{
\subsection{No $r$-violations imply no $r$-cross pairs.}\label{sec:noviol}
In this subsection we prove \Thm{noviol}. This closely follows the techniques and proofs from~\cite{ChSe13}.
\noviol*
This requires the alternating path setup of \cite{ChSe13}. 
 Recall the weight function $w(x,y) = \max(f(x) - f(y) - \pdi(x,y), f(y) - f(x) - \pdi(y,x))$ defined on pairs of the domain. 
 Note that $(x,y)$ is a violation iff $w(x,y) > 0$. 
 Let $M$ be a maximum weight minimum cardinality (MWmC) matching of $\VG(f,\cPd{\pdi})$ with the {\em minimum number} of $r$-cross pairs. 
 Recall an  $r$-cross pair $(x,y)$ has $x_r \neq y_r$.
 We will prove that this minimum value is $0$.
 Let $\cross(M)$ be the set of $r$-cross pairs in $M$. Let $\str(M) := M\setminus \cross(M)$. 
 For contradiction's sake, assume $\cross(M)$ is nonempty. Let $(x,y)\in \cross(M)$ be an arbitrary $r$-cross pair with $x_r = a$ and $y_r = b$ with $a \neq b$.
 Define matching $H := \{(u,v): u_r = a, v_r = b, u_j = v_j, j\neq i\}$. 
 This is a matching by projection between points with $r$th coordinate $a$ and $b$.
 For convenience, we denote the points with $r$th coordinate $a$ (resp. $b$) as the \emph{$a$-plane} (resp. $b$-plane).
 Consider the alternating paths and cycles in $H \Delta \;\str(M)$. The vertex $y$ is incident to only an $H$-pair, since $(x,y)\in \cross(M)$. 
 Let $y = s_1,s_2,\ldots,s_t$ be the alternating path starting from $y$, collectively denoted by $S$. 
We let $s_0 := x$. The end of $S$, $s_t$, may be either $M$-unmatched or $\cross(M)$-matched. 
In the latter case, we define $s_{t+1}$ to be such that $(s_t,s_{t+1})\in \cross(M)$. 
For even $i$, $(s_{i-1},s_i)$ is an $H$-pair and $(s_i,s_{i+1})$ is an $M$-pair. 
We list out some basic claims about $S$.
\begin{claim} \label{clm:alt} If strictly positive $j \equiv 0,1 \mod 4$, then $s_j$ is in the $b$-plane.
Otherwise, $s_j$ is in the $a$-plane.
\end{claim}
\begin{claim} \label{clm:w} For strictly positive even $i$, 
$f(s_{i-1}) - f(s_{i}) - \pdi(s_{i-1},s_{i}) \leq 0$ and $f(s_{i}) - f(s_{i-1}) - \pdi(s_{i},s_{i-1}) \leq 0$.
\end{claim}
\begin{proof} Since $f$ is $r$-good and $H$-pairs differ only in the $r$th coordinate, $w(s_{i-1},s_{i}) \leq 0$ for all even $i$.
The definition of $w(s_{i-1},s_{i})$ completes the proof.
\end{proof}
\begin{claim} \label{clm:d} For strictly positive $i\equiv 0\mod 4$, $\pdi(s_{i-1},s_i) = \pdi(s_2,s_1)$.
For $i\equiv 2\mod 4$, $\pdi(s_{i},s_{i-1}) = \pdi(s_2,s_1)$.
\end{claim}
\begin{proof} The point $s_0$ (which is $x$) lies in the $a$-plane.
Hence, for any $i\equiv 2\mod 4$, $s_i$ lies in the $b$-plane. Similarly, for $i \equiv 0\mod 4$, $s_i$ lies in the $a$-plan.
For strictly positive even $i$, $(s_{i-1},s_i)$ is an $H$-pair. An application of the projection property completes the proof.
\end{proof}
\begin{claim} \label{clm:d2} For strictly positive even $i$, $\pdi(s_i,s_{i+1}) = \pdi(s_{i-1},s_{i+2})$ and $\pdi(s_{i+1},s_i) = \pdi(s_{i+2},s_{i-1})$.
\end{claim}
\begin{proof} Consider $\str(M)$-pair $(s_i, s_{i+1})$. Both points are on the same ($a$ or $b$-)plane. 
Observe that $s_{i-1}$ is the projection of $s_i$ and $s_{i+2}$ is the projection of $s_{i+1}$ onto the other plane. Apply the projection property of $d$ to complete
the proof.
\end{proof}
Now we have all the ingredients to prove the theorem.
The strategy is to find another matching $M'$ such that either $w(M') > w(M)$ or  $w(M')=w(M)$ and $M'$ has strictly fewer cross pairs. 
Let us identify certain subsets of pairs to this end.
For even $k$, define
\[E_-(k) := (s_0,s_1), (s_2,s_3),\ldots, (s_k,s_{k+1}) = \{(s_j,s_{j+1}): j \textrm{ even, } 0\leq j\leq k\}\]
These are precisely the $\str(M)$-pairs in $S$ in the first $k$-steps. Note that $|E_-(k)| = k/2 + 1$. 
Now we define $E_+(k)$. In English: first pick pair $(s_0,s_2)$; subsequently pick the first unpaired $s_i$ and pair it with the next unpaired $s_j$ of the {\em opposite} parity. More precisely, for even $k$,
\[E_+(k) := (s_0,s_2), (s_1,s_4), (s_3,s_6),\ldots, (s_{k-3},s_k)  = (s_0,s_2) \cup \{(s_{j-3},s_j) : j \textrm{ even, }4\leq j\leq k\} \]
Note that $|E_+(k)| = k/2$.	
Wlog, assume that $w(x,y) = f(x) - f(y) - \pdi(x,y)$. 
It turns out the weights of all other $M$-pairs in $S$ are determined. We will assert that the pattern is as follows.
\begin{equation} \label{eq:cond} \tag{$\clubsuit$}
w(s_i,s_{i+1}) = 
\begin{cases} 
	f(s_i) - f(s_{i+1}) - \pdi(s_i,s_{i+1}) & \text{if } i\equiv 0\mod 4  \\
  f(s_{i+1}) - f(s_{i}) - \pdi(s_{i+1},s_{i}) & \text{if } i\equiv 2\mod 4
\end{cases}
\end{equation}
The following lemma determines the weights of all other $M$-edges in the alternating path $S$. Recall $(s_i,s_{i+1}) \in \str(M)$ for even $i$.
\begin{lemma}\label{lem:match} Suppose $s_i$ exists. If \Eqn{cond} holds for all even indices $< i$,
then $s_i$ is matched in $M$.
\end{lemma}
\begin{proof} Assume $i \equiv 2 \mod 4$. (The other case is analogous and omitted.) 
We prove by contradiction, so suppose $s_i$ is not matched in $M$. We set
$M' := M - E_-(i-2) + E_+(i)$. Note that $M'$ is a valid matching, since $s_i$ is not matched.
We compare $w(M')$ and $w(M)$. By \Eqn{cond}, we can express $w(E_-(i-2))$ exactly.
\begin{eqnarray}
w(E_-(i-2)) & = & \sum_{j: \textrm{even, } 0\leq j\leq i-2} w(s_j,s_{j+1}) \nonumber \\
& = & [f(s_0) - f(s_1) - \pdi(s_0,s_1)] + [f(s_3) - f(s_2) -\pdi(s_3,s_2)] + \nonumber \\
                  &    & [f(s_4) - f(s_5) - \pdi(s_4,s_5)]  + [f(s_7) - f(s_6) - \pdi(s_7,s_6)] + \cdots \nonumber \\
                  &     & [f(s_{i-2}) - f(s_{i-1}) - \pdi(s_{i-2},s_{i-1})] \label{eq:w-}
\end{eqnarray}
We lower bound $w(E_+(i))$. Since each individual weight term is a maximum of two expressions,
we can choose either. We set the expression up to match $w(E_-(i-2))$ as best as possible.
\begin{eqnarray}
w(E_+(i))   & \geq  & [f(s_0) - f(s_2) - \pdi(s_0,s_2)] + [f(s_4) - f(s_1) - \pdi(s_4,s_1)] + \nonumber \\
						&				& [f(s_3) - f(s_6) - \pdi(s_3,s_6)] + [f(s_8) - f(s_5) - \pdi(s_8,s_5)] + \nonumber \\
            &    		& [f(s_{i-3}) - f(s_{i}) - \pdi(s_{i-3},s_{i})] \label{eq:w+}
\end{eqnarray}
Note that $w(M') - w(M) = w(E_+(i)) - w(E_-(i-2))$. Observe that any $f$ term that occurs in both \Eqn{w-} and \Eqn{w+}
has the same coefficient. By \Clm{d2}, $\pdi(s_3,s_2) = \pdi(s_4,s_1)$, $\pdi(s_4,s_5) = \pdi(s_3,s_6)$, etc.
$$ w(E_+(i)) - w(E_-(i-2)) \geq f(s_{i-1}) - f(s_i) - \pdi(s_0,s_2) + \pdi(s_0,s_1) $$
The points $s_0$ and $s_1$ lie is different planes, and $(s_1,s_2) \in H$. We can apply the linearity property
to get $\pdi(s_0,s_1) = \pdi(s_0,s_2) + \pdi(s_2,s_1)$. Plugging this in, applying \Clm{w} and \Clm{d} for $i$,
$$ w(E_+(i)) - w(E_-(i-2)) \geq f(s_{i-1}) - f(s_i) + \pdi(s_2,s_1) = -[f(s_i) - f(s_{i-1}) - \pdi(s_i,s_{i-1})] \geq 0 $$
Hence $w(M') \geq w(M)$. Note that $|M'| - |M|$ $= |E_+(i)| - |E_-(i-2)| $ $= i/2 - ((i-2)/2 + 1) = 0$.
Finally, observe that $E_+(i)$ has no $r$-cross pairs, but $E_-(i-2)$ has one (pair $(s_0,s_1)$).
This contradicts the choice of $M$ as a MWmC matching with the least $r$-cross pairs.
\end{proof}
\begin{claim} \label{clm:distinct} If \Eqn{cond} holds for all even indices $< i$,
then $s_0, s_1, \ldots, s_{i+1}$ are all distinct.
\end{claim}
\begin{proof} (This is trivial if $i < t$. The non-trivial case if when $S$ ends as $s_i$.)
The points $s_1, \ldots, s_{i}$ are all distinct. If $s_i \neq x$, the claim holds.
So assume $s_i = x = s_0$. By \Clm{alt}, $i \equiv 2 \mod 4$. Replace pairs $A = \{(s_0,s_1), (s_{i-2},s_{i-1})\}$ by $(s_{i-2},s_1)$.
Note that $\pdi(s_0,s_1) = \pdi(s_0,s_{i-1}) + \pdi(s_{i-1},s_1)$.
By \Eqn{cond},
\begin{eqnarray*}
	w(A) & = & [f(s_0) - f(s_1) - \pdi(s_0,s_1)] + [f(s_{i-2}) - f(s_{i-1}) - \pdi(s_{i-2},s_{i-1})] \\
	& = & [f(s_{i-2}) - f(s_1) - \pdi(s_{i-2},s_{i-1}) - \pdi(s_{i-1},s_1)] + [f(s_0) - f(s_{i-1}) - \pdi(s_0,s_{i-1})] \\
	& \leq & [f(s_{i-2}) - f(s_1) - \pdi(s_{i-2},s_1)] \leq w(s_{i-2},s_1)
\end{eqnarray*}
The total number of pairs has decreased, so we complete the contradiction.
\end{proof}
\begin{lemma}\label{lem:cond} Suppose $s_i$ exists. If \Eqn{cond} holds for all even indices $< i$,
then \Eqn{cond} holds for $i$. 
\end{lemma}
\begin{proof} We prove by contradiction, so \Eqn{cond} is false for $i$.
(Again, assume $i \equiv 2 \mod 4$. The other case is omitted.)
By \Clm{distinct}, $E_+(i-2) \cup (s_{i-3},s_{i+1})$ is a valid set of matched pairs.
Let $M' := M- E_-(i) + (E_+(i-2) \cup (s_{i-3},s_{i+1}))$. 
Observe that $|M'| = |M| - 1$ and the vertices $s_{i-1}$ and $s_i$ are left unmatched in $M'$.
By \Eqn{cond} for even indices $< i$ and the opposite of \Eqn{cond} for $i$,
\begin{eqnarray}
w(E_-(i)) & = & [f(s_0) - f(s_1) - \pdi(s_0,s_1)] + [f(s_3) - f(s_2) -\pdi(s_3,s_2)] + \nonumber \\
                  &    & [f(s_4) - f(s_5) - \pdi(s_4,s_5)]  + [f(s_7) - f(s_6) - \pdi(s_7,s_6)] + \cdots \nonumber \\
                  &     & [f(s_{i-2}) - f(s_{i-1}) - \pdi(s_{i-2},s_{i-1})] + [f(s_{i}) - f(s_{i+1}) - \pdi(s_{i},s_{i+1})] \label{eq:w-2}
\end{eqnarray}
We stress that the last weight is ``switched".
We lower bound $w(E_+(i-2) \cup (s_{i-3},s_{i+1}))$ appropriately. 
\begin{eqnarray}
w(E_+(i-2) \cup (s_{i-3},s_{i+1}))   & \geq  & [f(s_0) - f(s_2) - \pdi(s_0,s_2)] + [f(s_4) - f(s_1) - \pdi(s_4,s_1)] + \nonumber \\
						&				& [f(s_3) - f(s_6) - \pdi(s_3,s_6)] + [f(s_8) - f(s_5) - \pdi(s_8,s_5)] + \cdots \nonumber \\
            &    		& [f(s_{i-7}) - f(s_{i-4}) - \pdi(s_{i-7},s_{i-4})] + [f(s_{i-2}) - f(s_{i-5}) - \pdi(s_{i-2},s_{i-5})] + \nonumber \\
            &				&	[f(s_{i-3}) - f(s_{i+1}) - \pdi(s_{i-3},s_{i+1})] \label{eq:w+2}
\end{eqnarray}
As before, we subtract \Eqn{w-2} from \Eqn{w+2}. All function terms from \Eqn{w+2} cancel out. By \Clm{d2}, all $\pdi$-terms
except the first and last cancel out.
$$ w(M') - w(M) \geq f(s_{i-1}) - f(s_i) - \pdi(s_0,s_2) - \pdi(s_{i-3},s_{i+1}) + \pdi(s_0,s_1) + \pdi(s_{i-2},s_{i-1}) + \pdi(s_i,s_{i+1}) $$
By linearity, $\pdi(s_0,s_1) = \pdi(s_0,s_2) + \pdi(s_2,s_1)$. Furthermore, by \Clm{d}, $\pdi(s_2,s_1) = \pdi(s_{i},s_{i-1})$.
By \Clm{d2}, $\pdi(s_{i-2},s_{i-1}) = \pdi(s_{i-3},s_i)$. By triangle inequality,
$-\pdi(s_{i-3},s_{i+1}) + \pdi(s_{i-3},s_i) + \pdi(s_i,s_{i+1}) \geq 0$. Putting it all together and applying \Clm{w},
$$w(M') - w(M) \geq -[f(s_i) - f(s_{i-1}) - \pdi(s_i,s_{i-1})] \geq 0$$
So $M'$ has at least the same weight but lower cardinality than $M$. Contradiction.
\end{proof}
\begin{lemma}\label{lem:str} Suppose $s_i$ exists. If \Eqn{cond} holds for all even indices $< i$,
then $s_i$ is matched in $\str(M)$. 
\end{lemma}
\begin{proof} Suppose not. (Again, assume $i \equiv 2 \mod 4$.) By \Lem{match}, $s_i$ is matched in $M$, so $(s_i,s_{i+1}) \in \cross(M)$.
We set $M' = M - E_-(i) + (E_+(i) \cup (s_{i-1},s_{i+1}))$. By \Clm{distinct}, $M'$ is a valid matching.
We have $|M'| = |M|$. $M$ has two $r$-cross pairs $(s_0,s_1)$ and $(s_i,s_{i+1})$, but $M'$ has at
most one $(s_{i-1},s_{i+1})$. It suffices to show that $w(M') \geq w(M)$ to complete the contradiction.
By \Lem{cond} and \Eqn{cond},
\begin{eqnarray}
w(E_-(i)) & = & [f(s_0) - f(s_1) - \pdi(s_0,s_1)] + [f(s_3) - f(s_2) -\pdi(s_3,s_2)] + \nonumber \\
                  &    & [f(s_4) - f(s_5) - \pdi(s_4,s_5)]  + [f(s_7) - f(s_6) - \pdi(s_7,s_6)] + \cdots \nonumber \\
                  &     & [f(s_{i-2}) - f(s_{i-1}) - \pdi(s_{i-2},s_{i-1})] + [f(s_{i+1}) - f(s_{i}) - \pdi(s_{i+1},s_{i})] \nonumber
\end{eqnarray}
\begin{eqnarray}
w(E_+(i) \cup (s_{i-1},s_{i+1}))   & \geq  & [f(s_0) - f(s_2) - \pdi(s_0,s_2)] + [f(s_4) - f(s_1) - \pdi(s_4,s_1)] + \nonumber \\
						&				& [f(s_3) - f(s_6) - \pdi(s_3,s_6)] + [f(s_8) - f(s_5) - \pdi(s_8,s_5)] + \cdots \nonumber \\
            &    		& [f(s_{i-3}) - f(s_{i}) - \pdi(s_{i-3},s_{i})] + [f(s_{i+1}) - f(s_{i-1}) - \pdi(s_{i+1},s_{i-1})] \nonumber 
\end{eqnarray}
All function terms and all but the first and last $\pdi$-terms cancel out.
The second inequality below holds by linearity and triangle inequality. The last equality is an application of
\Clm{d}.
\begin{eqnarray*}
w(M') - w(M) & \geq & \pdi(s_0,s_1) - \pdi(s_0,s_2) + \pdi(s_{i+1},s_i) - \pdi(s_{i+1},s_{i-1}) \\
& \geq & \pdi(s_2,s_1) - \pdi(s_i,s_{i-1}) = 0
\end{eqnarray*}
\end{proof}
Finally, we prove \Thm{noviol}.
\begin{proof} We started with a MWmC matching $M$ with the minimum number of $r$-cross pairs.
If there exists at least one such cross pair $(x,y)$, we can define the alternating path
sequence $S$. Wlog, we assumed \Eqn{cond} holds for $i=0$. Applications of \Lem{cond} and \Lem{str}
imply that $S$ can never terminate. Contradiction.
\end{proof}
}
\section{Search Trees and Bounded Derivative Property Testing.}\label{sec:ub}
As a result of dimension reduction, we can focus on designing testers for the line $[n]$. 
Our analysis is simple, but highlights the connection between bounded-derivative property testing and optimal search trees. 
\subsection{Testers for the Line $[n]$.}\label{sec:line-ub}
Let $T$ be {\em any} binary search tree (BST) with respect to the totally ordered domain $[n]$. 
Every node of $T$ is labeled with a unique entry in $[n]$, and the left (resp. right) child, if it exists, has a smaller (resp. larger) entry.
The {\em depth} of a node $v$ in the tree $T$, denoted as $\depth_T(v)$, is the number of {\em edges} on its path to the root. So the root has depth $0$.
Given a distribution $\cD$ on $[n]$, the expected depth of $T$ w.r.t. $\cD$ is denoted as $\Delta(T;\cD) = \Exp_{v\sim \cD}[\depth_T(v)]$. 
The depth of the optimal BST w.r.t. $\cD$ is denoted by $\Delta^*(\cD)$.
It has long been observed that the transitivity of violations is the key property required
for monotonicity testing on $[n]$~\cite{BRW05,EKK+00,ACCL07,JR11}. 
We distill this argument down
to a key insight:
Given {\em any} BST $T$, there exists the following \submit{one-sided} tester $\textrm{BST}(T)$ for $\cP$ on the line.
\begin{shaded}
{\bf BST Tester ($T$)}
\begin{asparaenum}
\itemsep0em 
\item Sample $v \sim \cD$.
\item If $v$ is the root of $T$, do nothing.
\item Else, query $f(u)$ for all vertices lying on the path from $v$ to root (including the root and $v$).
\item Reject if any pair of these vertices form a violation to $\cP$.
\end{asparaenum}
\end{shaded}
\full{It is clear that the tester never rejects a function satisfying $\cP$. }
(To connect with previous work, observe that the list of ancestor-descendant pairs forms a 
2-Transitive Closure spanner~\cite{BGJRW09}.)
\begin{lemma}\label{lem:tree-tester}
For any bounded derivative property $\cP$,  $\Pr[\textrm{BST tester rejects}]\geq \dist_\cD(f,\cP)$.
\end{lemma}
\begin{proof}
Let $X$ be the set of non-root nodes $v$ of $T$ with the following property: $(u,v)$ is a violation to $\cP$ for
some node $u$ on the path from $v$ to the root of $T$. The probability of rejection of the BST tester is precisely $\mu_\cD(X)$.
We claim that $X$ is a vertex cover of $\VG(f,\cP)$ which proves the lemma using \Lem{characterization}.
Pick any violation $(x,y)$ and assume without loss of generality $f(x) - f(y) > \pdi(x,y)$.
Let $z$ be the lowest common ancestor of $x$ and $y$ in $T$. By the BST property, either 
$x < z < y$ of $x>z>y$. By the linearity property of $\pdi$, we get $\pdi(x,y) = \pdi(x,z) + \pdi(z,y)$. This implies either $f(x) - f(z) > \pdi(x,z)$ or $f(z) - f(y) > \pdi(z,y)$, that is,
either $(x,z)$ or $(y,z)$ is a violation implying one of them is in $X$.
\end{proof}
\begin{restatable}{lemma}{optbst}
\label{thm:tree-optbst}\label{lem:optbst}
For any BST $T$, there is a $24\eps^{-1}\Delta(T;\cD)$-query line monotonicity-tester.
\end{restatable}
\begin{proof} 
The expected number of queries made by the BST tester is $\sum_{v: \textrm{ non-root}} \Pr[v]\cdot (\depth_T(v) + 1) = (1-\Pr[\textrm{root}]) + \Delta(T;\cD) \leq 2\cdot\Delta(T;\cD)$.
\full{\[\sum_{v: \textrm{ non-root}} \Pr[v]\cdot (\depth_T(v) + 1) = (1-\Pr[\textrm{root}]) + \Delta(T;\cD) \leq 2\cdot\Delta(T;\cD)\]}
The expected depth is at least $(1-\Pr[\textrm{root}])$ since non-roots have depth at least $1$.
To get a bonafide tester with deterministic query bounds, 
run the BST tester $2/\eps$ times, aborting (and accepting) if the total number of queries exceeds $24\Delta(T;\cD)/\eps$. 
The expected total number of queries is at most $4\Delta(T;\cD)/\eps$.
By Markov's inequality, the probability that the tester aborts is  $\leq 1/6$.  
By \Lem{tree-tester}, if $\dist_\cD(f;\cP) > \eps$,
the probability that this tester does not find a violation is at most $(1-\eps)^{2/\eps} \leq 1/6$. With probability $\geq (1-1/6-1/6) = 2/3$, the tester rejects an $\eps$-far function.
\end{proof}
Choose $T$ to be the optimal BST to get the following theorem.
\begin{theorem}\label{thm:ub-line-dist-known}
There exists a $24\eps^{-1}\Delta^*(\cD)$-query tester for any bounded derivative property over the line. 
\end{theorem}
Note that once the tree is fixed, the BST tester only needs random samples from the distribution.
Pick $T$ to be the balanced binary tree of depth $O(\log n)$ to get a {\em distribution-free} tester.
\begin{theorem}
There exists a $24\eps^{-1}\log n$-query {\em distribution free} tester for any  bounded-derivative property over the line. 
\end{theorem}
\subsection{Testers for the Hypergrid.}\label{sec:hg-ub}
Given a series of BSTs $T_1, T_2, \ldots, T_d$ corresponding to each dimension, we have 
the following hypergrid BST tester. 
\begin{shaded}
{\bf Hypergrid BST Tester ($T_1, T_2, \ldots, T_d$)}
\begin{asparaenum}
\itemsep0em 
\item Sample $x \sim \cD$.
\item Choose dimension $r$ u.a.r. and let $\ell$ be the $r$-line through $x$.
\item Run \textbf{BST Tester($T_r$)} on $\frest{f}{\ell}$.
\end{asparaenum}
\end{shaded}
\begin{lemma}\label{lem:hyper-bst} For any set of BSTs $T_1, T_2, \ldots, T_d$,
the probability of rejection is at least $\dist_\cD(f,\cP)/4d$.
\end{lemma}
\begin{proof} 
Condition on an $r$-line being chosen. The probability distribution over $r$-lines for this tester is $\cD_{-r}$.
By \Lem{tree-tester}, the rejection probability is at least 
$\Exp_{\ell \sim \cD_{-r}} [\dist_{\cD_r}(\frest{f}{\ell},\cP)] = \dist^r_{\cD}(f,\cP)$.
The overall rejection probability is at least $\sum_{r=1}^d \frac{\dist^i_{\cD}(f,\cP)}{d} \geq \frac{\dist_\cD(f,\cP)}{4d}$, by \Thm{dimred}. 
\end{proof}
\noindent
The expected number of queries made by this procedure is at most $\frac{1}{d}\!\cdot\!\sum_{r=1}^d \!2\Delta(T_r;\cD_r)$. 
Repeating it $O(d/\eps)$ times to get the desired tester. The proof of the following is identical to that of \Lem{optbst} and is omitted.
\begin{restatable}{lemma}{optbst-hg}
\label{thm:tree-optbst}\label{lem:optbst-hg}
For any collections of BSTs $(T_1,\ldots,T_d)$, there is a $100\eps^{-1}\sum_{i=1}^d\Delta(T;\cD)$-query tester for any bounded derivative property. 
\end{restatable}
\noindent
As in the case of the line we get the following as corollaries. 
\distknown*
\distfree*
The upper bound $\sum_{r=1}^d \Delta^*(\cD_r) $ is at most $H(\cD)$, but can be much smaller, and it is clearest in the case of the hypercube. In the hypercube, each $\cD_r$ is  given by $(\mu_r,1-\mu_r)$. Set $\theta_r := \min(\mu_r,1-\mu_r)$. The optimal BST places the point of larger mass on the root and has expected depth $\theta_r$.  
\distknowncube*
It is instructive to open up this tester.  It samples a point $x$ from the distribution $\cD$ and picks a dimension $r$ uniformly at random.
With probability $\theta_r$, it queries both endpoints of $(x,x\oplus \be_r)$. With probability $(1-\theta_r)$, it does {\em nothing}. This process is repeated $O(d/\eps)$ times.
When $\theta_r = \mu_r =1/2$, this is the standard edge tester. 
\section{Lower Bounds}\label{sec:lb}
We prove that the upper bounds of \Sec{ub} are tight up to the dependence on the distance parameter $\eps$. 
As alluded to in \Sec{ourtechs}, we can only prove lower bounds for {\em stable} product distributions.
These are distributions where small perturbations to the mass function do not change $\Delta^*$ drastically.
\begin{definition}[{\bf Stable Distributions}]\label{def:stable}
A product distribution $\cD$ is said to be $(\epsilon,\rho)$-stable if for all product distributions $\cD'$ with $||\cD - \cD'||_{\TV} \leq \eps$, 
$\Delta^*(\cD') \geq \rho\Delta^*(\cD)$. 
\end{definition}
The uniform distribution on $[n]^d$ is $(\eps,1-o(1))$-stable, for any constant $\eps < 1$.
The Gaussian distribution also shares the same stability. 
An example of an unstable
distribution is the following. Consider $\cD$ on $[n]$, where the probability on the first $k = \log n$ 
elements is $(1-\eps)/k$, and is $\eps/(n-k)$ for all other elements. Let $\cD'$ have all its mass uniformly spread
on the first $k$ elements. We have $||\cD - \cD'||_\TV = \epsilon$ but $\Delta^*(\cD) \approx \eps\log n$ and $\Delta^*(\cD') \approx \log k = \log\log n$.
\mainlowerbound*
\submit{
The first observation is that testing monotonicity can be reduced to testing bounded-derivative property. 
\begin{theorem} \label{thm:red} Fix domain $[n]^d$ and a distribution $\cD$.
Suppose there exists a $Q$-query tester for testing a bounded-derivative property $\cP$  with distance parameter $\epsilon$. Then there
exists a $(Q + 10/\epsilon)$-query tester for monotonicity for functions $f:[n]^d\mapsto\NN$ over $\cD$ with distance parameter $2\epsilon$.
\end{theorem}
Note that the distribution $\cD$ may possibly be unknown in the distribution-free setting. We sketch the construction.
Given a function $f$ for which we need to test monotonicity, and given a bounded derivative property $\cP$, 
we construct a function $g$ defined as $g(x) = \frac{\delta}{2R}f(x) - \pdi(\bzero,x)$. Here $R$ is the largest value $f$ can take, and $\delta$ is a certain `slack' parameter
dependent on $\pdi$. One can show that the violation graphs of $f$ with respect to monotonicity, and $g$ with respect to $\cP$ are {\em identical}, and so, a tester for $\cP$ can be used to test monotonicity of $f$. 
Hence, it suffices to prove lower bounds for monotonicity. We use the proof strategy set 
up in~\cite{E04, ChSe13-2} that reduce general testers to comparison-based testers.
We encapsulate the main approach in the following theorem, proven implicitly
in~\cite{ChSe13-2}. (We use $\MON$ to denote the monotonicity property.)
\begin{theorem} \label{thm:lb-frame} 
Fix domain $[n]^d$, product distribution $\cD$, proximity parameter $\eps$,
and positive integer $L$ possibly depending on $\cD$ and $\eps$. A pair $(x,y)$ distinguishes function $g$ from $h$ if $h(x) < h(y)$ and $g(x) > g(y)$.
A set $Q \subset [n]^d$ distinguishes $g$ from $h$ if there exists $(x,y) \in Q$ doing so.
Suppose there is a collection of `hard' functions $h,g_1,\ldots,g_L: [n]^d \mapsto \NN$ such that
\begin{asparaitem}
	\itemsep-0.1em
	\item The function $h$ is monotone.
	\item Every $\dist_\cD(g_i,\MON) \geq \epsilon$.
	\item Any set $Q \subset [n]^d$ distinguishes at most $|Q|$ of the $g_i$'s from $h$.
\end{asparaitem}
Then any (even adaptive, two-sided) monotonicity tester w.r.t. $\cD$ for functions $f:[n]^d\mapsto\NN$ with distance parameter $\epsilon$ must make $\Omega(L)$ queries.
\end{theorem} 
In \Sec{lb-line} we first describe hard functions for $[n]$, which cements the connection between testing and search trees. 
The general proof of \Thm{the-lower-bound}  is more complicated and we give a sketch in \Sec{lb-intuit}.  
}
\full{
\subsection{Reduction from monotonicity to bounded-derivative property}\label{sec:bdp-to-mono}
Consider a function $f: [n]^d \mapsto [R]$ with where $R \in \NN$. 
Let $\pdi$ be the distance function obtained by bounding family $\B$.
We let $\bzero \in [n]^d$ be $(0,0,\ldots,0)$. 
We use $\prec$ to denote the natural partial order in $[n]^d$, and let $\hcd(x,y)$ be the highest
common descendant of $x,y \in [n]^d$.
We first prove an observation about triangle equality.
\begin{observation}\label{obs:equal} If $\pdi(\bzero,x) + \pdi(x,y) = \pdi(\bzero,y)$,
then $x \prec y$.
\end{observation}
\begin{proof} By linearity, $\pdi(x,y) = \pdi(x,\hcd(x,y)) + \pdi(\hcd(x,y),y)$.
Since $\hcd(x,y) \prec x$, by linearity again, 
$\pdi(\bzero,x) = \pdi(\bzero,\hcd(x,y)) + \pdi(\hcd(x,y),x)$. (Similarly for $y$.)
Putting it all into the `if' condition,
\begin{align*}
\pdi(\bzero,\hcd(x,y)) + \pdi(\hcd(x,y),x) + \pdi(x,\hcd(x,y)) + \pdi(\hcd(x,y),y)
	= \pdi(\bzero,\hcd(x,y)) + \pdi(\hcd(x,y),y)
\end{align*}
This yields $\pdi(\hcd(x,y),x) + \pdi(x,\hcd(x,y)) = 0$. Suppose $\hcd(x,y) \neq x$.
The length (in terms of $\B$)
of the path from $\hcd(x,y)$ to $x$ involves a sum of $u_i(t)$ terms, and the reverse
path involves corresponding $-l_i(t)$ terms. Since $u_i(t) > l_i(t)$, the total path length
from $\hcd(x,y)$ to $x$ and back is strictly positive. Therefore, $\hcd(x,y) = x$ and $x \prec y$.
\end{proof}
Let $U$ be the set of incomparable (ordered) pairs in $[n]^d$.
Define $\delta := \min_{(x,y) \in U} \{\pdi(\bzero,x)+\pdi(x,y) - \pdi(\bzero,y)\}$. 
By \Obs{equal}, $\delta > 0$.
Define
\[ g(x) := \frac{\delta}{2R}\cdot f(x) - \pdi(\bzero,x) \]
\begin{lemma}\label{lem:bijective}
$\dist_\cD(g,\cP) =\dist_\cD(f,\MON)$.
\end{lemma}
\begin{proof} We show that $(u,v)$ violates $\cPd{\pdi}$ of $g$ iff it violates monotonicity of $f$.
First, the `only if' case.
Assume $g(u)-g(v)>\pdi(u,v)$. Plugging in the expression for $g(\cdot)$ and rearranging,
\begin{eqnarray*}
\frac{\delta}{2R}(f(u)-f(v))> \pdi(\bzero,u) + \pdi(u,v) - \pdi(\bzero,v)
\end{eqnarray*}
By triangle inequality on the RHS, $f(u) > f(v)$.
Note that $f(u)-f(v)\le R$ so $\frac{\delta}{2R}(f(u)-f(v))\leq \delta/2$. 
So $\delta/2 > \pdi(\bzero,u) + \pdi(u,v) - \pdi(\bzero,v)$. By choice of $\delta$, the RHS must be zero.
By \Obs{equal}, $u \prec v$, and $(u,v)$ is a violation to monotonicity of $f$.
Now the `only if' case, so $u \prec v$ and $f(u)>f(v)$. Note that $\pdi(\bzero, v)=\pdi(\bzero,u)+\pdi(u,v)$. 
We deduce that $(u,v)$ is also a violation to $\cPd{\pdi}$ for $g$.
$$g(u)-g(v)=\frac{\delta}{2R}(f(u)-f(v))+\pdi(\bzero,v)-\pdi(\bzero,u)=\frac{\delta}{2R}(f(u)-f(v))+\pdi(u,v)>\pdi(u,v)$$
\end{proof} 
Our main reduction theorem is the following.
\begin{theorem} \label{thm:red} Fix domain $[n]^d$ and a product distribution $\cD$.
Suppose there exists a $Q$-query tester for testing a bounded-derivative property $\cP$  with distance parameter $\epsilon$. Then there
exists a $Q + 10/\epsilon$-query tester for monotonicity for functions $f:[n]^d\mapsto\NN$ over $\cD$ with distance parameter $2\epsilon$.
\end{theorem}
\begin{proof} The monotonicity tester first queries $10/\eps$ points of $[n]^d$,
each i.i.d. from $\cD$.
Let the maximum $f$-value among these be this $M$. 
Consider the truncated function $f':[n]^d \mapsto [M]$, where $f'(x) = M$ if $f(x) \geq M$
and $f'(x) = f(x)$ otherwise. If $f$ is monotone, $f'$ is monotone.
Note that $\dist_\cD(f,f') < \eps$. So if $f$ is $2\eps$-far from monotone,
$f'$ is $\eps$-far from monotone. We can apply the $\cP(\B)$ tester
on the function $g$ obtained from \Lem{bijective}.
\end{proof}
\subsection{Monotonicity Lower Bound Framework.}\label{sec:monotone-lb-framework}
The lower bound for monotonicity testing goes by the proof strategy set 
up in~\cite{ChSe13-2}. This is based on arguments in~\cite{E04,ChSe13-2} that 
reduce general testers to comparison-based testers.
We encapsulate the main approach in the following theorem, proven implicitly
in~\cite{ChSe13-2}. (We use $\MON$ to denote the monotonicity property.)
\begin{theorem} \label{thm:lb-frame} 
Fix domain $[n]^d$, distribution $\cD$, proximity parameter $\eps$,
and positive integer $L$ possibly depending on $\cD$ and $\eps$. A pair $(x,y)$ distinguishes function $g$ from $h$ if $h(x) < h(y)$ and $g(x) > g(y)$.
Suppose there is a collection of `hard' functions $h,g_1,\ldots,g_L: [n]^d \mapsto \NN$ such that
\begin{asparaitem}
	\itemsep-0.1em
	\item The function $h$ is monotone.
	\item Every $\dist_\cD(g_i,\MON) \geq \epsilon$.
	\item Pairs in any set $Q \subset [n]^d$, can distinguish at most $|Q|$ of the $g_i$'s from $h$.
\end{asparaitem}
Then any (even adaptive, two-sided) monotonicity tester w.r.t. $\cD$ for functions $f:[n]^d\mapsto\NN$ with distance parameter $\epsilon$ must make $\Omega(L)$ queries.
\end{theorem}
In \Sec{lb-line} and \Sec{lb-cube},
we first describe hard functions for the line and the hypercube domain, respectively. 
The general hypergrid is addressed in \Sec{lb-hg}. 
}
\subsection{The Line} \label{sec:lb-line}
\begin{theorem}\label{thm:lb-line}
Fix a parameter $\eps$. If $\cD$ is $(2\eps,\rho)$-stable, then any $\epsilon$-monotonicity tester w.r.t. $\cD$ for functions $f:[n]^d\mapsto\NN$ requires $\Omega(\rho\Delta^*(\cD))$ queries.
\end{theorem}
Not surprisingly, the lower bound construction is also based on BSTs. 
We specifically use the \emph{median BST}~\cite{Melhorn75}. 
When $n=1$, then the tree is the singleton. For a general $n$, 
let $t \in [n]$ be the smallest index such that $\mu(\set{1,\cdots,t}) \geq 1/2$ (henceforth, in this section, we use $\mu$ to denote $\mu_\cD$). 
The root of $T$ is $t$. Recur the construction on the intervals $[1,t-1]$ and $[t+1,n]$.
By construction, 
the probability mass of any subtree together with its parent is greater than the probability mass of the 
sibling subtree. This {\em median property} will be utilized later.
We follow the framework of \Thm{lb-frame} to construct a collection of hard functions. 
The monotone function $h$ can be anything; $h(i) = 3i$ works. We will construct a function $g_j$ ($j \geq 1$) for each non-root level of the median BST. 
Consider the nodes at depth $j-1$ (observe the use of $j-1$, and not $j$). Each of these corresponds to an interval, and we denote this sequence
of intervals by $\inter{j}{1}, \inter{j}{2}, \ldots$. (Because internal nodes of the tree are also elements in $[n]$,
there are gaps between these intervals.) Let $L_{\geq j} := \{x: \depth_T(x)\geq j\}$ be the nodes at depth $j$ and higher.
We have the following simple claim.
\begin{claim}\label{clm:line-interval}
$\inter{j}{k}$ can be further partitioned into $\inter{j}{k,\lleft}$ and $\inter{j}{k,\rright}$ such that
$\sum_k \min\left(\mu(\inter{j}{k,\lleft}), \mu(\inter{j}{k,\rright})\right)\!\geq\! \frac{\mu(L_{\geq j})}{2}.$
\end{claim}
\begin{proof} Consider the node $u_k$ corresponding $\inter{j}{k}$, and let the nodes in the left and right subtrees
be $S_\ell$ and $S_r$. If $\mu(S_\ell) \leq \mu(S_r)$, then $\inter{j}{k,\lleft} = S_\ell\cup u_k$ 
and $\inter{j}{k,\rright} = S_r$. Otherwise, $\inter{j}{k,\lleft} = S_\ell$ and $\inter{j}{k,\rright} = u_k \cup S_r$.
By the median property of the BST, $\min(\mu(\inter{j}{k,\lleft}), \mu(\inter{j}{k,\rright})) = \max(\mu(S_\ell),\mu(S_r))$ $\geq (\mu(S_\ell) + \mu(S_r))/2$.
\end{proof}
We describe the non-monotone $g_j$'s and follow up with some claims. Let $\lca(x,y)$ denote the least common ancestor of $x$ and $y$ in $T$.
\begin{equation}
\label{eq:try} g_{j}(x) = 
\left\{ 
\begin{array} {l l l}
2x \quad \textrm{if $x \notin \bigcup_k \inter{j}{k}$}\\
2x + 2(b-m) + 1 \quad \textrm{if $x \in \inter{j}{k,\lleft} = [a,m]$, ~~where $\inter{j}{k} = [a,b]$.}\\
2x - 2(m-a) - 1 \quad \textrm{if $x \in \inter{j}{k,\rright} = [m+1,b]$, where $\inter{j}{k} = [a,b]$.}
\end{array} \right.
\end{equation}
\noindent
\begin{restatable}{claim}{gj}
\label{clm:gj} 
	{\sf (i)} $\distmon{g_j}{\cD}\!\geq\!\frac{\mu(L_{\geq j})}{2}$.
	{\sf (ii)} If $(x,y)$ distinguishes $g_j$ from $h$, then $\lca(x,y)$ lies in level $(j-1)$. 
\end{restatable}
\begin{proof} All elements in $\inter{j}{k,\lleft}$ are in violation with all elements in $\inter{j}{k,\rright}$ for all $k$.
To see this, let $x\in \inter{j}{k,\lleft}$ and $y\in \inter{j}{k,\rright}$, and so $x \prec y$.  Denote $\inter{j}{k} = [a,b]$,
$$g_j(x) - g_j(y) = 2x + 2(b-m) + 1 - 2y + 2(m-a) +1  = 2(x-a) +2(b-y) + 2 > 0$$
The vertex cover of the violation graph of $g_i$ has mass at least $\sum_k \min(\mu(\inter{j}{k,\lleft}),\mu(\inter{j}{k,\rright})) \geq \mu(L_{\geq j})/2$ (\Clm{line-interval}). 
This proves part (i). To prove part (ii), let $x \prec y$ distinguish $g_j$ from $h$, so $g_j(x) > g_j(y)$. We claim there exists a $k^*$ such that 
$x \in \inter{j}{k^*,\lleft}$ and $y\in \inter{j}{k^*,\rright}$. 
For any $\inter{j}{k} = [a,b]$, the $g_j$ values lie in $[2a+1,2b+1]$. Hence, if $x \in \inter{j}{k}$
and $y \notin \inter{j}{k}$ (or vice versa), $(x,y)$ is not a violation.
So $x$ and $y$ lie in the same $\inter{j}{k^*}$, But the function restricted to $\inter{j}{k^*,\lleft}$
or $\inter{j}{k^*,\rright}$ is increasing, completing the proof.
\end{proof}
\noindent
The following claim is a simple combinatorial statement about trees.
\begin{restatable}{claim}{lcaa}\label{clm:lca}
Given a subset $Q$ of $[n]$, let $\lca(Q) = \set{\lca(x,y) : x, y \in Q}$. Then $|\lca(Q)| \leq |Q| - 1$.
\end{restatable}
\full{
\begin{proof}
The proof is by induction on $|Q|$. The base case of $|Q| = 2$ is trivial. Suppose $|Q| > 2$. 
Consider the subset $P \subseteq Q$ of all elements of $Q$, none of whose ancestors are in $Q$.
Also observe that if $P = Q$, then $\lca(Q)$ are precisely the internal nodes of a binary tree whose leaves are $Q$, and therefore $|\lca(Q)| \leq |Q|-1$.
If $P$ is a singleton, then $\lca(Q) = \lca(Q \setminus P) + 1 \leq |Q \setminus P| -1 + 1 = |Q|-1$ (inequality from induction hypothesis).
So assume $P \subset Q$ and $|P| \neq 1$.
For $p \in P$, let $S_p$ be the set of elements of $Q$ appearing in the tree rooted at $p$. 
For every $x \in S_p$ and $y \in S_{p'}$ ($p \neq p'$), $\lca(x,y) = \lca(p, p')$.
Furthermore, the sets $S_p$ non-trivially partition $Q$.
 Therefore, $\lca(Q) = \lca(P) \cup \bigcup_{p \in P}\lca(S_p)$.  Applying the induction hypothesis, 
 $|\lca(Q)| \leq |P| - 1 + \sum_{p \in P} |S_p| - |P| = |Q| - 1$. 
\end{proof}
}
\noindent
Let $\ell_\eps$ be the largest $\ell$ such that $\mu(L_{\geq \ell}) \geq 2\eps$. By \Clm{gj}.(i), the collection of functions $\{g_1,\ldots,g_{\scriptscriptstyle \ell_\eps}\}$ are each $\eps$-far from monotone.  By \Clm{gj}.(ii) and \Clm{lca}, a subset $Q \subseteq [n]$ can't distinguish more than $|Q|$ of these functions from $h$.
 \Thm{lb-frame} gives an $\Omega(\ell_\eps)$ lower bound and \Thm{lb-line} follows from \Clm{lepsilon_is_large}.
\begin{claim}\label{clm:lepsilon_is_large}
$\ell_\epsilon \geq \rho\Delta^*(\cD)$.
\end{claim}
\begin{proof}
Consider the distribution $\cD'$ that transfers all the mass from $L_{\geq \ell_\eps+1}$ to the remaining vertices proportionally. 
That is, if $\nu:=\mu(L_{\geq \ell_\eps+1})$, then  $\mu_{\cD'}(i) = 0$ for $i\in L_{\geq \ell_\eps+1}$, and $\mu_{\cD'}(i) = \mu_{\cD}(i)/(1-\nu)$ for the rest. 
Observe that $||\cD - \cD'||_{\TV} = \mu_{\cD}(L_{\geq \ell_\eps+1}) < 2\eps$. Also observe that since $T$ is a binary tree of height $\ell_\eps$, $\ell_\eps \geq \Delta^*(\cD')$: the LHS is the max depth, the RHS is the (weighted) average depth. Now, we use stability of $\cD$. Since $\cD$ is $(2\eps,\rho)$-stable, $\Delta^*(\cD') \geq \rho\Delta^*(\cD)$. 
\end{proof} 
\full{
\subsection{The Boolean hypercube}\label{sec:lb-cube}
For the boolean hypercube, the lower bound doesn't require the stability assumption. 
Any product distribution over $\{0,1\}^d$ is determined by the $d$ fractions $(\mu_1,\ldots,\mu_d)$, where $\mu_r$ is the probability of $0$ on the $r$-th coordinate. Let $\theta_r := \min(\mu_r,1-\mu_r)$. 
\begin{theorem}\label{thm:lb-cube}
Any monotonicity tester w.r.t. $\cD$ for functions $f:\{0,1\}^d\mapsto\NN$ with distance parameter $\epsilon\leq 1/10$ must make $\Omega\left(\sum_{r=1}^d \min(\mu_r,1-\mu_r)\right)$ queries.
\end{theorem}
We begin with the basic setup. A tester for non-trivial $\eps$ makes
at least $1$ query, so we can assume that $\sum_{r=1}^d \theta_r > 1$. 
For ease of exposition, assume $\theta_r = \mu_r$, for all $1\leq r \leq d$. 
(If not, we need to divide into two cases depending on $\theta_r$ and argue analogously for each case.)
Assume wlog $\theta_1\leq \theta_2 \leq \cdots \leq \theta_d$. 
Partition $[d]$ into contiguous segments $I_1,\ldots,I_b,I_{b+1}$ such that for each $1\leq a\leq b$,  $\sum_{r\in I_a}\theta_r \in [1/2,1)$. Observe that $b = \Theta\left(\sum_r\theta_r\right)$.
For $1\leq a\leq b$, define the indicator functions $\chi_a:\{0,1\}^d\mapsto\{0,1\}$ as follows:
$$ \chi_a(x) = 
\left\{ 
\begin{array} {l l}
1& \quad \textrm{if $\forall i \in I_a, x_i=1$}\\
0& \quad \textrm{otherwise \quad ($\exists i \in I_a, x_i = 0$)}
\end{array} \right.
$$
By \Thm{lb-frame}, we need to define the set of functions with appropriate properties.
The monotone function $h(\cdot)$ is defined as $h(x) = \sum_{a=1}^b \chi_a(x)2^a$. The functions $g_1,\ldots,g_b$ are defined as 
$$ g_a(x) = 
\left\{ 
\begin{array} {l l}
h(x) - 2^{r} - 1& \quad \textrm{if $\chi_a(x) = 1$}\\
h(x) & \quad \textrm{if $\chi_a(x) = 0$}
\end{array} \right.
$$
We prove all the desired properties.
\begin{claim}
For all $a$, $\dist_\cD(g_a,\MON) \geq 1/10$.
\end{claim}
\begin{proof} Let $I$ denote $I_a$, and $J = [n] \setminus I$.
Think of $x = (x_{I}, x_{J})$. Fix $\v$ in $\{0,1\}^{|J|}$,
and define sets $X_1(\v) := \{x | \chi_a(x) = 1, x_{J}=\v\}$ (a singleton) and 
$X_0(\v) = \{x | \chi_a(x) = 0, x_{J}=\v\}$. 
Note that $\bigcup_\v(X_1(\v)\cup X_0(\v))$ forms a partition of the cube.
For $c \neq a$, $\chi_c(x)$ is the same for all $x \in (X_0(\v) \cup X_1(\v))$. Hence, for {\em any} $x \in X_0(\v)$ and $y \in X_1(\v)$,
$x \prec y$ and $g_a(x) > g_a(y)$.
Any vertex cover in the violation graph must contain either $X_1(\v)$ or $X_0(\v)$, for each $\v$. 
Let $\cD_I$ be the conditional distribution on the $I$-coordinates. 
In the following, we use the inequalities $\sum_{i \in I}\theta_i \in [1/2,1)$ and $1 - t \in [e^{-2t},e^{t}]$ for $t\leq 1/2$.
\begin{eqnarray*}
\mu_{\cD_I}(X_1(\v)) & = & \prod_{i\in I}(1-\theta_i) \geq \exp(-2\sum_{i\in I}\theta_i) \geq e^{-2} > 1/10 \\
\mu_{\cD_I}(X_0(\v)) & = & 1 - \prod_{i\in I}(1-\theta_i) \geq 1 - \exp(-\sum_{i\in I}\theta_i) \geq 1 - e^{-1/2} > 1/10.
\end{eqnarray*}
For each $\v$, the conditional mass of the vertex cover is at least $1/10$, and therefore, the $\mu_\cD$ mass of the vertex cover is at least $1/10$.
\end{proof}
A pair $x,y$ in $[n]^d$ \emph{captures} index $a$
if $a$ is the largest index such that $\chi_a(x) \neq \chi_a(y)$. Furthermore, a set $Q$ captures $a$ if it contains a pair capturing $a$.
\begin{restatable}{claim}{hccap}
\label{clm:hccap}
If $Q$ distinguishes $g_a$ from $h$, then $Q$ must capture $a$.
\end{restatable}
\begin{proof} Consider  $x, y \in Q$ where $h(x) < h(y)$ but $g_a(x) > g_a(y)$.
It must be that $\chi_a(x) = 0$ and $\chi_a(y) = 1$. Suppose this pair does not capture $a$.
There must exist index $c > a$ (let it be the largest) such that $\chi_c(x) \neq \chi_c(y)$. Because $h(x) < h(y)$, 
$\chi_c(x) = 0$ and $\chi_c(y) = 1$. 
By definition, $g_a(y) - g_a(x) = (h(y) -2^a - 1) - h(x)$.
We have $h(y) - h(x) = \sum_{t=1}^c (\chi_t(y) - \chi_t(x))2^t$.
Since $\chi_a(x) = \chi_c(x) = 0$ and $\chi_a(y) = \chi_c(y) = 1$,
$h(y) - h(x) \geq 2^c + 2^a - \sum_{t < c:t\neq a} 2^t$.
Combining, 
$$ g_a(y) - g_a(x) \geq 2^c + 2^a - \sum_{t < c:t\neq a} 2^t  - 2^a - 1 = 2^a > 0.$$
\end{proof}
\begin{claim}\label{clm:cap}[Lifted from~\cite{ChSe13-2}.]
 A set $Q$ captures at most $|Q|-1$ coordinates.
\end{claim}
\begin{proof} We prove this by induction on $|Q|$. When $|Q| = 2$, this is trivially true.
Otherwise, pick the largest coordinate $j$ captured by $Q$ and let $Q_0 =\{x: x_j = 0\}$ and $Q_1 = \{x:x_j = 1\}$. 
By induction, $Q_0$ captures at most $|Q_0|-1$ coordinates, and $Q_1$ captures at most $|Q_1| -1$ coordinates.
Pairs $(x,y)\in Q_0\times Q_1$ only capture coordinate $j$. The total number of captured coordinates is at most $|Q_0| - 1 + |Q_1| - 1 + 1 = |Q| - 1$.
\end{proof}
We can now invoke \Thm{lb-frame} to get an $\Omega(b)=\Omega(\sum_r\theta_r)$ lower bound thereby proving \Thm{lb-cube}.
\medskip
The hypercube lower bound can be generalized to give a weak lower bound for hypergrids,
which will be useful for proving the stronger bound.
Fix a dimension $r$. For any $1\leq j\leq n$, define $\theta^j_r := \min(\sum_{k\leq j} \mu_{\cD_r}(k), 1 - \sum_{k\leq j} \mu_{\cD_r}(k))$. Define $\theta_r := \max_{1\leq j\leq n} \theta^j_r$.
Note that $\theta_r$ generalizes the above definition for the hypercube. The following theorem follows by a reduction to the hypercube lower bound. 
\begin{restatable}{theorem}{hgpart}
\label{thm:hg-part1}\label{thm:lb-hg-part1}
Any monotonicity tester on the hypergrid with distance parameter $\eps\leq 1/10$, makes
$\Omega\left(\sum_{r=1}^d \theta_r\right)$ queries.
\end{restatable}
\begin{proof}
For $1\leq r \leq d$, let $1\leq j_r\leq n$ be the $j$ such that $\theta_r = \theta^j_r$. Project the hypergrid onto a Boolean hypercube using the following mapping $\psi:[n]^d\to\{0,1\}^d$: for $x\in [n]^d$, $\psi(x)_r = 0$ if $x_r \leq j_r$, and $1$ otherwise. The corresponding product distribution $\cD'$ on the hypercube puts $\mu_{\cD'_r}(0) = \sum_{k\leq j_r}\mu_{\cD_r}(k)$, for all $r$. Note that $\min(\mu_r,1-\mu_r) = \theta_r$. Given any function $f$ on $\{0,1\}^d$, extend it to $g$ over the hypergrid in the natural way: for $x\in [n]^d$, $g(x) = f(\psi(x))$. 
Note that $\dist_{\cD'}(f,\MON) = \dist_{\cD}(g,\MON)$. (This is akin to \Thm{bhg}.) 
Any tester for $g$ over $[n]^d$ induces a tester for $f$ on $\{0,1\}^d$ with as good a query complexity: whenever the hypergrid tester queries $x\in [n]^d$, the hypercube tester queries $\psi(x)$. Therefore, the lower bound \Thm{lb-cube} for the hypercube implies \Thm{hg-part1}.
\end{proof}
}
\full{
\subsection{The Hypergrid.}\label{sec:lb-hypergrid}\label{sec:lb-hg}
Our main lower bound result is the following, which implies \Thm{the-lower-bound} via \Thm{red}.
\begin{theorem}\label{thm:lb-hg}
For any parameter $\epsilon < 1/10$, 
and for any $(\const\eps,\rho)$-stable, product distribution $\cD$,
any (even adaptive, two-sided) montonicity tester w.r.t. $\cD$ for functions $f:[n]^d \mapsto \NN$ with proximity parameter $\eps$
requires $\Omega(\rho\Delta^*(\cD))$ queries.
\end{theorem}
}
\full{\subsubsection{The intuition}\label{sec:lb-intuit}}
\submit{\subsection{Intuition for the Hypergrid Proof}\label{sec:lb-intuit}}
Since we already have a proof for $d=1$ in \Sec{lb-line}, an obvious
approach to prove \Thm{the-lower-bound} is via some form of induction on the dimension.
Any of the
$g_j$-functions on $[n]$ in \Sec{lb-line} can be extended the obvious way to a function on $[n]^d$.
Given (say) $g_j:[n] \mapsto \NN$, we can define $f:[n]^d \mapsto \NN$ as $f(x) = g_j(x_1)$.
Thus, we embed the hard functions for $\cD_1$ along dimension $1$. One can envisage a way do the same for dimension $2$, 
and so on and so forth, thereby leading to $\sum_i \Delta^*(\cD_i)$ hard functions in all.
There is a caveat here. The construction of \Sec{lb-line} for (say) $\cD_1$ requires
the stability of $\cD_1$. Otherwise, we don't necessarily get $\Omega(\Delta^*(\cD_1))$ functions
with distance at least $\epsilon$.
For instance, if the root of the median BST has more than $(1-\epsilon)$ fraction of  the weight, we get at most one hard function of distance at least $\epsilon$.
So, the above approach requires stability of all the \emph{marginals} of $\cD$. Unfortunately, there exist stable product
distributions with all marginals unstable. Consider $\cD = \prod_r \cD_r$, where
each $\cD_r = (\frac{1}{(n-1)d},\ldots,\frac{1}{(n-1)d}, 1-\frac{1}{d})$.
Note that $\Delta^*(\cD) \approx \log n$.
Each $\cD_r$ is individually  unstable (for $\eps > 1/d$), since there is a $\cD'_i$ with all the mass on the $n$th coordinate, such that $\|\cD_i - \cD'_i\|_{TV} = 1/d$ and $\Delta^*(\cD'_i) = 0$.
On the other hand, it is not hard to see that $\cD$ is $(1/100,1/100)$-stable. 
A new idea is required to construct the lower bound.
\full{To see this, suppose there is a product distribution $\cD'$
such that $\Delta^*(\cD') < \Delta^*(\cD)/100 = (\log n)/100$. Markov's inequality implies that
for $\Omega(d)$ dimensions, $\|\cD'_r - \cD_r\|_{\TV} = \Omega(1/d)$. A calculation shows that $\|\cD - \cD'\|_{\TV}$
must be at least $1/100$. In sum, for any constants $\eps,\rho$, there exist $(\eps,\rho)$-stable distributions $\cD$ such that
each marginal $\cD_r$ is only $(\eps/d,\rho)$-stable.
This is a major roadblock for a lower bound construction, and therefore
a new idea is required.}
\full{We design an \emph{aggregation} technique that does the following. Start with 1D functions $g^1_{j_1}$ and $g^2_{j_2}$
that are hard functions from \Sec{lb-line} for $\cD_1$ and $\cD_2$ respectively. Suppose the corresponding
distances to monotonicity are $\eps^{(1)}$ and $\eps^{(2)}$. We construct a function $f:[n]^d \mapsto \NN$
that is $\eps^{(1)} + \eps^{(2)}$-far, so we can effectively add their distances. If we can aggregate
$\Omega(d)$ 1D functions, each with distance $\eps/d$, then we get a desired hard function.
As can be expected, this construction is quite delicate, because we embed violations in many dimensions
simultaneouly. Furthermore, we need to argue that this aggregation can produce enough ``independent"
hard functions, so we get a large enough lower bound (from \Thm{lb-frame}). And that is where the hard work lies.
}
\submit{Our way out of the impasse is to {\em aggregate} hard functions from different dimensions in such a way that their distances to monotonicity add up.
To explain this better, we introduce some notation. Call the median BST for $\cD_r$ as $T_r$, and let us use $\Delta_r$ to denote the expected depth of $T_r$ w.r.t. $\cD_r$.
Abusing notation, let $\Delta(\cD) = \sum_{r=1}^d \Delta_r$. 
Let $\dep{r}{j} := \{x: \depth_{T_r}(x) \geq j\}$.  We use the shorthand $\below{r}{j}$ to denote $\mu_{\cD_i}(\dep{r}{j})$. 
The discussion of the previous paragraph points to the issue that $\below{r}{j}$ can be $\ll \epsilon$ for $j \approx \Delta_r$. 
However, the global sum of all these quantities is large enough, as following simple but crucial claim asserts.
This follows since $\Exp[Z] = \sum_{k\in \N}\Pr[Z\geq k]$, for any non-negative, integer valued random variable.
\begin{claim} \label{clm:search-depth} 
$\sum_r \sum_{j\geq 1} \below{r}{j} = \sum_r \Exp_{x\sim \cD_r}[\depth_{T_r}(x)] = \Delta(\cD)$.
\end{claim}
The aggregation step picks a collection of hard functions $g^r_{j_r}$s as constructed in \Sec{lb-line},
each corresponding to the layer $j_r$ of $T_r$, with at most one function per tree, 
such that $\sum_{r ~\textrm{picked}}\below{r}{j_r} \approx \epsilon$. 
We form a single hard function from such a collection whose distance to monotonicity is of the 
order of this sum. 
The stability of $\cD$ is crucially used (with \Clm{search-depth}) to argue that $\Omega(\Delta^*(\cD))$ such
`independent' collections exist, so the framework of \Thm{lb-frame} applies.
This technical details are quite subtle and can be found in the full version of the paper.
}
\full{
\subsubsection{Setup and Construction} 
Fix $\epsilon$ and let $\epsilon' = \const\epsilon$. Fix 
the $(\epsilon',\rho)$-stable distribution $\cD$. 
Since $\cD$ is $(\epsilon',\rho)$-stable, for any $\cD'$ with $\|\cD'-\cD\|_{\TV} \leq \epsilon'$, we have 
$\Delta^*(\cD') \geq \rho\Delta^*(\cD)$. 
We denote the median BST for $\cD_r$ as $T_r$, $\Delta_r$ as the expected depth w.r.t. $\cD_r$, and $\Delta(\cD) = \sum_{r=1}^d \Delta_r$. 
The following shows that the median BST is near optimal.
}
\full{
\begin{lemma} \label{lem:median} For any product distribution $\cD = \prod_r \cD_r$, $\Delta(\cD) \leq 5\Delta^*(\cD)$.
\end{lemma}
\begin{proof} 
Fix a coordinate $r$. The depth of a vertex $u$ in $T_r$ is at most $\log_2(1/\mu_{\cD_r}(u))$, so we get $\Delta_r \leq H(\cD_r)$, the Shannon entropy of $\cD_r$.
It is also known (cf. Thm 2 in~\cite{Melhorn75}) that $H(\cD_r) \leq \log_23(\Delta^*(\cD_r) + 1)$. To see this, notice that any BST can be converted into a prefix-free ternary code
of expected length $(\Delta^*(\cD_r) +1)$, say, over the alphabet `left',`right', and `stop'.
Therefore, if $\Delta^*(\cD_r) \geq 1/2$, we have $\Delta_r \leq 5\Delta^*(\cD_r)$.
If $\Delta^*(\cD_r) < 1/2$, then since $\Delta^*(\cD_r) \geq 1 - \Pr[\textrm{root}]$, 
we get $\mu^* := \mu_{\cD_r}(u^*) > 1/2$ where $u^*$ is the root of the optimal BST $T^*$. But this implies $u^*$ is also the root of $T_r$ by construction of the median BST. Now we can prove via induction.
If $p$ and $q$ are the total masses of the nodes in the left and right sub-tree of $T^*$ (and therefore also $T_r$), and $\Delta^*_1$ (resp. $\Delta_1$) and $\Delta^*_2$ (resp. $\Delta_2$)be the expected depths of these subtrees in $T^*$ (resp. $T_r$), then we get, $\Delta^*(\cD_r) = p\Delta^*_1  + q\Delta^*_2 + (1-\mu^*)  \leq 5p\Delta_1 + 5q\Delta_2 + (1-\mu^*) \leq 5\Delta_r$.
\end{proof}
}
\noindent
\full{
\Thm{lb-frame} requires the definition of a monotone function and a collection of $\eps$-far from monotone
functions with additional properties. The monotone function is $\val(x) := \sum_{r=1}^d 2(2n+1)^r x_r$.
The non-monotone functions (which we refer to as ``hard" functions) are constructed via aggregation. 
From \Sec{lb-line}, for each dimension $r$ and each level $j \geq 1$ in tree $T_r$,
we have a 1D ``hard" function $\hard{r}{j}:[n] \mapsto \NN$. 
It is useful to abstract out some of the properties of $\hard{r}{j}$ that were
proved in \Sec{lb-line}.
Let $\lev{r}{j}$ be the nodes in $T_r$ at level $j$. Each level corresponds to a collection of intervals of $[n]$. 
We use $\dep{r}{j} := \bigcup_{j' \geq j} \lev{r}{j'}$ and $\abo{r}{j} = \bigcup_{j' < j} \lev{r}{j'}$.
We use the shorthand $\below{r}{j}$ to denote $\mu_{\cD_r}(\dep{r}{j})$. The following lemma is a restatement of \Clm{line-interval} and \Clm{gj}.
\begin{lemma} \label{lem:hard-prop} Consider $\hard{r}{j}:[n] \mapsto \NN$, for $j \geq 1$
	All violations to monotonicity are contained in intervals corresponding to $\lev{r}{j-1}$, and the distance to monotonicity is at least $\below{r}{j}/2$. 
	Furthermore, any violation $(x,y)$ has $\lca(x,y)$ in $\lev{r}{j-1}$.
\end{lemma}
}
\full{
The aggregation process takes as input a \emph{map} $\psi:[d] \mapsto \{\bot\} \cup \{2,3,4,\ldots\}$. Note that 
if $\psi(r) \neq \bot$, then $\psi(r) > 1$. Informally, $\psi(r)$, when not equating to $\bot$, tells us the level of $T_r$ whose hard function is to be included in the aggregation.
We define $\mappos := \{r | \psi(r) \neq \bot\}$, the subset of relevant dimensions.
Given the map $\psi$, we aggregate the collection of 1D functions $\{\hard{r}{\psi(r)} | r \in \mappos\}$ into a single hard function for $[n]^d$ as follows.
\begin{equation}
\label{eq:gg}
\gg(x) := \sum_{r \in \mappos} (2n+1)^r \hone{r}{\psi(r)}(x_r) + \sum_{r \notin \mappos} 2(2n+1)^r x_r
\end{equation}
Observe that the latter sum is identical to the corresponding portion in $\val(x)$.
The first summand takes the hard function corresponding to the $\psi(r)$th level of $T_r$ for $r\in \mappos$ and aggregates them via multiplying them with a suitable power of $(2n+1)$.
\begin{definition} \label{def:psi} A map $\psi$ is \emph{\bf useful} if the following are true.
\begin{asparaitem}
	\item $\sum_{r \in \mappos} \below{r}{\psi(r)} \in (\eps',1)$
	\item For all $r \in \mappos$, $\below{r}{\psi(r)} \geq \frac{\below{r}{\psi(r)\!-\!1}}{2}$.
\end{asparaitem}
\end{definition}
\noindent
In plain English, the first point states that total distance of the hard functions picked should be at least $\epsilon'$. The second point is a technicality which is required to argue about the  distance of the aggregated function. It states that in each relevant $T_r$, the total mass on the nodes lying in the $\psi(r)$th layer and below shouldn't be much smaller than the total mass on the nodes lying on the $(\psi(r)\!-\!1)$th layer and below.
}
\full{
\begin{lemma}\label{lem:dist-to-mon} If $\psi$ is useful, $\dist_\cD(\gg,\MON) \geq \eps$.
\end{lemma}
\begin{proof} It is convenient to consider restrictions of $\gg$ where
all coordinates in $[d] \setminus \mappos$ are fixed. This gives rise 
to $|\mappos|$-dimensional functions. We argue that each such restriction
is $\eps$-far from monotone, which proves the lemma.
Abusing notation, we use $\gg$ to refer to an arbitrary such restriction. 
Fix some $r \in \mappos$. Define the subset $S_r := \{x\in [n]^d: x_s\in \abo{s}{\psi(s)-1}, ~\forall s\neq r\}$ to be 
the set of points $x$ with the $s$th coordinate appearing in the first $(\psi(s)-2)$ layers of the tree $T_s$, for all $s\neq r$.
We stress that this is well-defined because $\psi(s) \geq 2$ by definition of $\psi$. 
Note that each $S_r$ is a collection of $r$-lines and
the restriction of $\gg$ on each line exactly a multiple
of $\hard{r}{\psi(r)}$. 
By \Lem{hard-prop}, all violations to monotonicity in such lines lie in the intervals corresponding to $\dep{r}{\psi(r)-1}$,
and the mass of the vertex cover of the violation graph (restricted to the line)
is at least $\below{r}{\psi(r)}/2$. Thus the total contribution to distance of $\gg$ from $S_r$ is at least $\frac{\below{r}{\psi(r)}}{2}\cdot  \mu_{\cD_{-r}}(\prod_{s \neq r} \abo{s}{\psi(s)-1})$.
What is crucial to note is that the regions of violations in $S_r$ is {\em disjoint} from the regions of violation in $S_{r'}$ for $r'\neq r$.
Therefore, the contributions to the distance of $\gg$ add up, and this gives 
\begin{eqnarray*}
	\dist_\cD(\gg,\MON) & \geq & \frac{1}{2} \sum_{r \in \mappos} \below{r}{\psi(r)} \cdot \mu_{\cD_{-r}}\Big(\prod_{s \neq r} \abo{s}{\psi(s)-1}\Big) \\
	& =  & \frac{1}{2} \sum_{r \in \mappos} \below{r}{\psi(r)} \prod_{s \neq r} (1 - \below{s}{\psi(s)-1}) \\
	& \geq &\frac{1}{2}  \sum_{r \in \mappos} \below{r}{\psi(r)} \prod_{s \neq r} (1 - 2\below{s}{\psi(s)}) 
	\ \ \ \ \textrm{(point 2 in def. of useful map)}
\end{eqnarray*}
We can apply the bound, $\sum_{r \in \mappos} \below{r}{\psi(r)} \in (\eps',1)$,
since $\psi$ is useful.
We lower bound the product by $\exp(-4\sum_{s \neq r} \below{s}{\psi(s)})$,
which by the above bound, is at least $e^{-4}$.
So, $\dist_\cD(\gg,\MON) \geq \sum_{r \in \mappos} \below{r}{\psi(r)}/\const \geq \eps'/\const = \epsilon$.
\end{proof}
}
\full{
\begin{definition}\label{def:disj}
Two maps $\psi_1, \psi_2$ are \emph{disjoint} if: $\{(r,\psi_1(r)) | r \in \Psi_1^{-1}\}$
and $\{(r,\psi_2(r)) | r \in \Psi_2^{-1}\}$ are disjoint.
\end{definition}
\noindent
That is, for every tree $T_r$, $\psi_1$ and $\psi_2$ point to different layers of the tree (or they point to $\bot$).
\begin{lemma} \label{lem:disting} Consider a set of maps $\psi_1, \psi_2, \ldots$
that are all pairwise disjoint. A set of $Q$ queries can distinguish at most 
$|Q|-1$ of these functions from $\val$.
\end{lemma}
\begin{proof} Say a pair $(x,y)$ of queries captures the (unique) tuple $(r,j)$ if the largest coordinate in which $x$ and $y$ differ is $r$, and furthermore $\lca(x_r,y_r)$ in $T_r$ lies in level $(j-1)$. A set $Q$ captures $(r,j)$ if some pair in $Q$ captures $(r,j)$.
We first show that if $(x,y)$ distinguishes $\gg$ from $\val$ for some map $\psi$, then $(x,y)$ captures
a pair $(r,\psi(r))$ for some $r \in \mappos$.
Assume wlog $\val(x) < \val(y)$, and so $\gg(x) > \gg(y)$.
Let $a$ be the largest coordinate at which $x$ and $y$ differ; since $\val(x) < \val(y)$, we get $x_a < y_a$. Suppose $\hard{a}{\psi(a)}(x_a)$ and $\hard{a}{\psi(a)}(y_a)$ is not a violation.
By the construction, this implies that $\hard{a}{\psi(a)}(y_a) - \hard{a}{\psi(a)}(x_a) \geq 1$.
Furthermore, $\hard{r}{\psi(r)}$ is always in the range $[1,2n]$ for any $r$.
\vspace{-3mm}
\begin{eqnarray*}
\gg(y) -\gg(x) & = & (2n+1)^a(\hard{a}{\psi(a)}(y_a) - \hard{a}{\psi(a)}(x_a)) 
+ \sum_{r < a, r\in \mappos} (2n+1)^r(\hard{r}{\psi(r)}(y_r) - \hard{r}{\psi(r)}(x_r)) \\
		    & \geq & (2n+1)^a - 
		    (2n)\sum_{r<a} (2n+1)^r \\
		   &  = & (2n+1)^a - (2n)\cdot \frac{(2n+1)^a - 1}{2n} ~~~~ > 0,
\end{eqnarray*}
So $(\hard{a}{\psi(a)}(x_a),\hard{a}{\psi(a)}(y_a))$ is a violation. Immediately,
we deduce that $\psi(a) \neq \bot$, so $a \in \mappos$. By \Clm{gj},
$\lca(x_a,y_a)$ lies in level $\psi(a) - 1$ of $T_a$, and hence, $(x,y)$ captures $(a,\psi(a))$.
As we prove in \Clm{qlca}, $Q$ queries can capture at most $|Q|-1$ such tuples.
The proof is completed by noting the maps $\psi_1,\psi_2,\ldots$ are pairwise disjoint.
\end{proof}
}
\full{
\begin{restatable}{claim}{qlca}\label{clm:qlca}
A nonempty set $Q$ can only capture at most $|Q|-1$ tuples $(r,j)$.
\end{restatable}
\begin{proof}
Proof is by induction on $|Q|$. 
If $|Q|=2$, then the claim trivially holds. Assume $|Q| > 2$.
Let $s$ be the largest dimension such that there are at least two points in $Q$ differing in that dimension.
For $c=1$ to $n$, let $Q_c := \{x\in Q: x_s = c\}$. 
By definition, $Q_c \subset Q$. 
Reorder the dimensions such that $Q_c$ is non-empty for $c = 1 \ldots q \leq n$.
By induction, each $Q_c$ captures at most $|Q_c|-1$ pairs for $1\leq c\leq q$. 
Consider $(x,y)$ with $x\in Q_c$ and $y\in Q_{c'}$ for $c \neq c'$.
The largest coordinate where they differ is exactly $s$. All tuples
captured by such pairs is of the form $(s,\ell)$, where
$\ell$ is the $\lca$ in $T_s$ of some $c,c' \in \{1\ldots,q\}$.
By \Clm{lca}, the total number of such points is at most $q-1$.
Thus, the total number of tuples captured is at most  $\sum_{a=1}^q|Q_a| - q + (q-1) = |Q|-1$.
\end{proof}
}
\full{
\subsubsection{Constructing the maps}
Let us go back to the framework of \Thm{lb-frame}.
From \Lem{dist-to-mon} and \Lem{disting}, it suffices to construct a sequence
$\psi_1, \psi_2,\ldots$ of pairwise disjoint, useful maps. The number
of such maps will exactly be our lower bound. The exact construction
is a little tricky, since the conditions of usefulness are somewhat
cumbersone. 
We use the following definition.
\begin{asparaitem}
	\item \textbf{Allowed levels:} A level $j$ is allowed w.r.t. dimension $r$ if $j > 1$ and $\below{r}{j} \geq \below{r}{j-1}/2$. This is in lines with point 2 of the usefulness definition.
	\item \textbf{Level sets $A_r$:} $A_r$ is the set of allowed levels of tree $T_r$.
\end{asparaitem}
}
\full{
It is convenient to define an abstract procedure that constructs these maps. 
We have a stack $S_r$ for each $r \in [d]$, whose elements are allowed levels.
The stack $S_r$ is initialized with $A_r$ in increasing order, that is the head (top entry) of the stack is the least (that is, closest to root) level
in $A_r$. In each \emph{round}, we will construct a map $\psi$. 
Denote the head of $S_r$ by $h_r$. Note that $h_r > 1$ by definition of allowed levels. 
Maintain a running count initialized to $0$.
We go through the stacks
in an arbitrary order popping off a \emph{single} element from each stack.
In a round, we never touch the same stack more than once.
When we pop $S_r$, we set $\psi(r) := h_r$ and add $\below{r}{h_r}$ to the running count.
We stop as soon as  the running count enters the interval $[\eps',1]$. 
For all $r$ for which $\psi(r)$ hasn't been defined, we set $\psi(r) = \bot$. This completes the description of a single map.
Observe, by definition of allowed levels and the stopping condition, $\psi$ is useful.
\def\PPsi{\mathbf{\Psi}}
When $\sum_{r=1}^d \below{r}{h_r} < \eps'$, we cannot complete the construction. So the procedure terminates, discarding the final map.
Let the set of maps constructed be $\PPsi$. By construction, the maps are useful. Furthermore, they are pairwise disjoint,
because once a layer is popped out, it never appears again.
We now basically show that $|\PPsi|$ is large, using the $(\eps',\rho)$-stability of $\cD$. This proves that the number of hard functions is large. 
We have to first deal with an annoying corner case of $\cD$.
}
\full{
\begin{theorem} \label{thm:heavy} If $\sum_r \below{r}{1} > \rho \Delta(\cD)/12$,
then any monotonicity tester requires $\Omega(\rho\Delta^*(\cD))$ queries.
\end{theorem}
\begin{proof} We simply apply the hypercube lower bound.
Recall the definition of $\theta_r$ described before \Thm{lb-hg-part1}. 
Note that $\below{r}{1}$
is simply the total $\cD_r$-mass of everything in $T_r$ other than the root.
By the median property of the $T_r$, $\theta_r$ is ensured to be at least 
half of this mass, and hence $\theta_r \geq \below{r}{1}/2$.
Combining with \Thm{lb-hg-part1}, we get a lower bound of $\Omega(\sum_r \below{r}{1})$,
which by assumption, is $\Omega(\rho\Delta(\cD))$. An application of \Lem{median} completes
the proof.
\end{proof}
}
\full{
Now we come to the main bound of $|\gcol|$.
We need some setup for the proof. The following simple observation is crucial. This follows since 
$\Exp[Z] = \sum_{k\in \N}\Pr[Z\geq k]$, for any non-negative, integer valued random variable.
\begin{claim} \label{clm:search-depth} 
For all $r$, $\sum_{j\geq 1} \below{r}{j} = \Exp_{x\sim \cD_r}[\depth_{T_r}(x)] = \Delta_r$.
\end{claim}
The following lemma completes the entire lower bound.
}
\full{
\begin{restatable}{lemma}{countx}
\label{lem:countX} Suppose $\sum_r \below{r}{1} \leq \rho \Delta(\cD)/12$.
Then $|\gcol|=\Omega(\rho\Delta(\cD))$.
\end{restatable}
\begin{proof}
Let $h_r$ denote the head of $S_r$ when the procedure terminates. So, $\sum_{r=1}^d \below{r}{h_r} < \eps'$.
For any $\psi \in \PPsi$,
$\sum_{r \in \gset} \below{r}{\psi(r)} \leq 1$. 
Hence, $|\gcol|$ is at least
the total sum over popped elements $\below{r}{j}$. Writing this out and expanding
out a summation,
\begin{eqnarray*}
|\gcol| \geq \sum_{r \in [d]}\sum_{j < h_r,j\in A_r} \below{r}{j} 
& = & \sum_{r \in [d]} \Big[\sum_{j = 1}^{h_r-1} \below{r}{j} - \below{r}{1} - \sum_{1 < j < h_r: j \notin A_r} \below{r}{j} \Big] 
\end{eqnarray*}
Recall that $h_r > 1$ and so the summations are well-defined.
For any level $1 < j \notin A_r$, we have $\below{r}{j} < \below{r}{j-1}/2$. Therefore,
$\sum_{1 < j < h_r: j \notin A_r} \below{r}{j} < \sum_{j = 1}^{h_r-2} \below{r}{j}/2$.
Plugging this bound in and applying the lemma assumption,
\begin{eqnarray}
\label{eq:toprove}
|\gcol| \geq \sum_{r \in [d]} \sum_{j = 1}^{h_r-1} \below{r}{j}/2 - \sum_{r \in [d]} \below{r}{1}
\geq \sum_{r \in [d]} \sum_{j = 1}^{h_r-1} \below{r}{j}/2 - \rho \Delta(\cD)/12
\end{eqnarray}
We need to lower bound the double summation above. Observe that the second summation is $\sum_{j\geq 1}\below{r}{j} - (\below{r}{h_r} + \below{r}{h_r+1} + \cdots )$.
The first term, by \Clm{search-depth} is precisely $\Delta(\cD)$, and by definition of $h_r$, each of the terms in the parenthesis is at most $\eps'$. However, the number of terms in the parenthesis can be quite large, and this doesn't seem to get any lower bound on the summation. Here's where stability of $\cD$ saves the day.
Construct a distribution $\cD'_r$ be the distribution on $[n]$ as follows. Move the entire probability mass away from $\dep{r}{h_r}$ 
and distribute it on the ancestral nodes in level $(h_r-1)$ of $T_r$. More precisely, $\mu_{\cD'_r}(u) = 0$ if $u \in \dep{r}{h_r}$, 
$\mu_{\cD'_r}(u) = \mu_{\cD_r}(u)$ if $u\in \abo{r}{h_r - 1}$, and $\mu_{\cD'_r}(u) = \sum_{v} \mu_{\cD_r}(v)$ for $u\in \lev{r}{h_r-1}$ where the summation is over children $v$ of $u$ in $T_r$.
Letting $\cD' := \prod_r \cD'_i$, we see  that $||\cD-\cD'||_{\TV}  \leq \sum_{r=1}^d \below{r}{h_r} < \eps'$. 
Since $\cD$ is $(\eps',\rho)$-stable, we get $\Delta^*(\cD') \geq \rho\Delta^*(\cD)$.
Now we can apply \Clm{search-depth} on $\cD'_r$ and $T_r$ to get 
$\EX_{x \sim \cD'_r}[\depth_{T_r}(x)] = \sum_{j\geq 1}\mu_{\cD'_r}(\dep{r}{j}) = \sum_{j=1}^{h_r - 1} \below{r}{j}.$
This expected depth is by definition at least $\Delta^*(\cD'_r)$. Therefore, we get a lower bound of $\sum_{r=1}^d \Delta^*(\cD'_r)/2 = \Delta^*(\cD')/2$ on the double summation in \Eqn{toprove}.
Using the stability of $\cD'$ this is at least $\rho\Delta^*(\cD)/2$. Substituting we get
\begin{eqnarray*}
	|\gcol|	& \geq & \rho\Delta^*(\cD)/2 - \rho \Delta(\cD)/12 \\ 
		        &  \geq & \rho\Delta^*(\cD)/2 - 5\rho\Delta^*(\cD)/12 = \Omega(\rho\Delta^*(\cD)) ~~~~~~ \textrm{(by \Lem{median})}
\end{eqnarray*} \end{proof}
We put it all together to prove the main lower bound, \Thm{lb-hg}.
If $\sum_r \below{r}{1} > \rho \Delta(\cD)/12$, \Thm{heavy} proves \Thm{lb-hg}. 
Otherwise, by \Lem{countX} we have constructed $\Omega(\rho\Delta^*(\cD))$ pairwise disjoint, useful maps. Each map yields a hard function of distance at least $\eps$ (by 
\Lem{dist-to-mon}), and these functions satisfy the conditions of \Thm{lb-frame}, which implies \Thm{lb-hg}.
}
\submit{
\bibliographystyle{alpha}
{\small
\begin{spacing}{0.9}
\bibliography{derivative-testing}
\end{spacing}
}}
\full{
\bibliographystyle{alpha}
\begin{spacing}{0.9}
\bibliography{derivative-testing}
\end{spacing}}
\end{document}